\documentclass[a4paper]{article}

\usepackage{amsmath}
\usepackage{latexsym}
\usepackage{amssymb}
\usepackage{graphicx}
\usepackage{amsthm}

\newtheorem{thm}{Theorem}[section]

\newtheorem{cor}[thm]{Corollary}
\newtheorem{pro}[thm]{Proposition}

\theoremstyle{definition}

\begin{document}
\title{Hamiltonian Cycles in Linear-Convex Supergrid Graphs}

\author{Ruo-Wei Hung\\
\textit{Department of Computer Science and Information Engineering,}\\
\textit{Chaoyang University of Technology,}\\
\textit{Wufeng, Taichung 41349, Taiwan}\\
E-mail: rwhung@cyut.edu.tw}

%%% ----------------------------------------------------------------------
\maketitle
%%% ----------------------------------------------------------------------

\begin{abstract}
A supergrid graph is a finite induced subgraph of the infinite graph associated with the two-dimensional supergrid. The supergrid graphs contain grid graphs and triangular grid graphs as subgraphs. The Hamiltonian cycle problem for grid and triangular grid graphs was known to be NP-complete. In the past, we have shown that the Hamiltonian cycle problem for supergrid graphs is also NP-complete. The Hamiltonian cycle problem on supergrid graphs can be applied to control the stitching trace of computerized sewing machines. In this paper, we will study the Hamiltonian cycle property of linear-convex supergrid graphs which form a subclass of supergrid graphs. A connected graph is called $k$-connected if there are $k$ vertex-disjoint paths between every pair of vertices, and is called locally connected if the neighbors of each vertex in it form a connected subgraph. In this paper, we first show that any 2-connected, linear-convex supergrid graph is locally connected. We then prove that any 2-connected, linear-convex supergrid graph contains a Hamiltonian cycle.\\

\end{abstract}

%%% ----------------------------------------------------------------------
\section{Introduction}
%%% ----------------------------------------------------------------------

A \textit{Hamiltonian cycle} in a graph is a simple cycle in which each vertex of the graph appears exactly once. The \textit{Hamiltonian cycle } problem involves determining whether or not a graph contains a Hamiltonian cycle. A graph is said to be \textit{Hamiltonian} if it contains a Hamiltonian cycle. The Hamiltonian path problem is defined similarly. They have numerous applications in different areas, including establishing transport routes, production launching, the on-line optimization of flexible manufacturing systems \cite{Ascheuer96}, computing the perceptual boundaries of dot patterns \cite{OCallaghan74}, pattern recognition \cite{Bermond78, PreperataS85, Toussaint80}, and DNA physical mapping \cite{Grebinski98}. It is well known that the Hamiltonian problems are NP-complete for general graphs \cite{GareyJ79, Johnson85}. The same holds true for bipartite graphs \cite{Krishnamoorthy76}, split graphs \cite{Golumbic80}, circle graphs \cite{Damaschke89}, undirected path graphs \cite{BertossiB86}, planar bipartite graphs with maximum degree 3 \cite{Itai82}, grid graphs \cite{Itai82}, triangular grid graphs \cite{Gordon08}, and supergrid graphs \cite{Hung15}.

Let $G$ be a graph with the vertex set $V(G)$ and edge set $E(G)$. Graph $G$ is \textit{connected} if there exists a path between every pair of its vertices, and $G$ is called $k$-connected ($k\geqslant 2$) if there are $k$ vertex-disjoint paths between every pair of vertices in $G$. Obviously, if graph $G$ is not $2$-connected, then $G$ contains no Hamiltonian cycle. In addition, a $2$-connected graph does not imply that it is Hamiltonian. For each vertex $v$ of $G$, the \textit{neighborhood} $N(v)$ of $v$ is the set of all vertices adjacent to $v$. For a subset of vertices $S\subseteq V(G)$, the subgraph of $G$ induced by $S$ is denoted by $G[S]$. A vertex $v$ of $G$ is said to be \textit{locally connected} if $G[N(v)]$ is connected. Graph $G$ is called \textit{locally connected} if each vertex of $G$ is locally connected. The locally connected property of graphs can be applied to VLSI architecture \cite{Gulak88}. Some further survey for local connectivity can be found in \cite{Faudree97}.

The \emph{two-dimensional integer grid} $G^\infty$ is an infinite graph whose vertex set consists of all points of the Euclidean plane with integer coordinates and in which two vertices are adjacent if and only if the (Euclidean) distance between them is equal to 1. The \emph{two-dimensional triangular grid} $T^\infty$ is an infinite graph obtained from $G^\infty$ by adding all edges on the lines traced from up-left to down-right. A fragment of graph $G^\infty$ and $T^\infty$ is depicted in Fig. \ref{Fig_FragmentOfGridRelated}(a) and Fig. \ref{Fig_FragmentOfGridRelated}(b), respectively. A \textit{grid graph} is a finite, vertex-induced subgraph of $G^\infty$. For a node $v$ in the plane with integer coordinates, let $v_x$ and $v_y$ be the $x$ and $y$ \textit{coordinates} of node $v$, respectively, denoted by $v=(v_x, v_y)$. If $v$ is a vertex in a grid graph, then its possible neighbor vertices include $(v_x, v_y-1)$, $(v_x-1, v_y)$, $(v_x+1, v_y)$, and $(v_x, v_y+1)$. For example, Fig. \ref{Fig_ExampleOfGridRelated}(a) shows a grid graph. A \textit{triangular grid graph} is a finite, vertex-induced subgraph of $T^\infty$. If $v$ is a vertex in a triangular grid graph, then its possible neighbor vertices include $(v_x, v_y-1)$, $(v_x-1, v_y)$, $(v_x+1, v_y)$, $(v_x, v_y+1)$, $(v_x-1, v_y-1)$, and $(v_x+1, v_y+1)$. For example, Fig. \ref{Fig_ExampleOfGridRelated}(b) shows a triangular grid graph. Thus, triangular grid graphs contain grid graphs as subgraphs. Note that triangular grid graphs defined above are isomorphic to the original triangular grid graphs studied in the literature \cite{Gordon08} but these graphs are different when considered as geometric graphs. By the same construction of triangular grid graphs from grid graphs, we have proposed a new class of graphs, namely \textit{supergrid graphs}, in \cite{Hung15}. The \emph{two-dimensional supergrid} $S^\infty$ is an infinite graph obtained from $T^\infty$ by adding all edges on the lines traced from up-right to down-left. A \emph{supergrid graph} is a finite, vertex-induced subgraph of $S^\infty$. The possible adjacent vertices of a vertex $v=(v_x, v_y)$ in a supergrid graph include $(v_x, v_y-1)$, $(v_x-1, v_y)$, $(v_x+1, v_y)$, $(v_x, v_y+1)$, $(v_x-1, v_y-1)$, $(v_x+1, v_y+1)$, $(v_x+1, v_y-1)$, and $(v_x-1, v_y+1)$. Then, supergrid graphs contain grid graphs and triangular grid graphs as subgraphs. For example, Fig. \ref{Fig_FragmentOfGridRelated}(c) depicts a fragment of graph $S^\infty$ and Fig. \ref{Fig_ExampleOfGridRelated}(c) shows a supergrid graph. Notice that grid and triangular grid graphs are not subclasses of supergrid graphs, and the converse is also true: these classes of graphs have common elements (points) but in general they are distinct since the edge sets of these graphs are different. Obviously, all grid graphs are bipartite \cite{Itai82} but triangular grid graphs and supergrid graphs are not bipartite. The Hamiltonian cycle problem for grid graphs and triangular grid graphs were known to be NP-complete \cite{Gordon08, Itai82}. Recently, we have shown that the Hamiltonian cycle problem on supergrid graphs is also NP-complete \cite{Hung15}.

\begin{figure}[t]
\begin{center}
\includegraphics[width=0.95\textwidth]{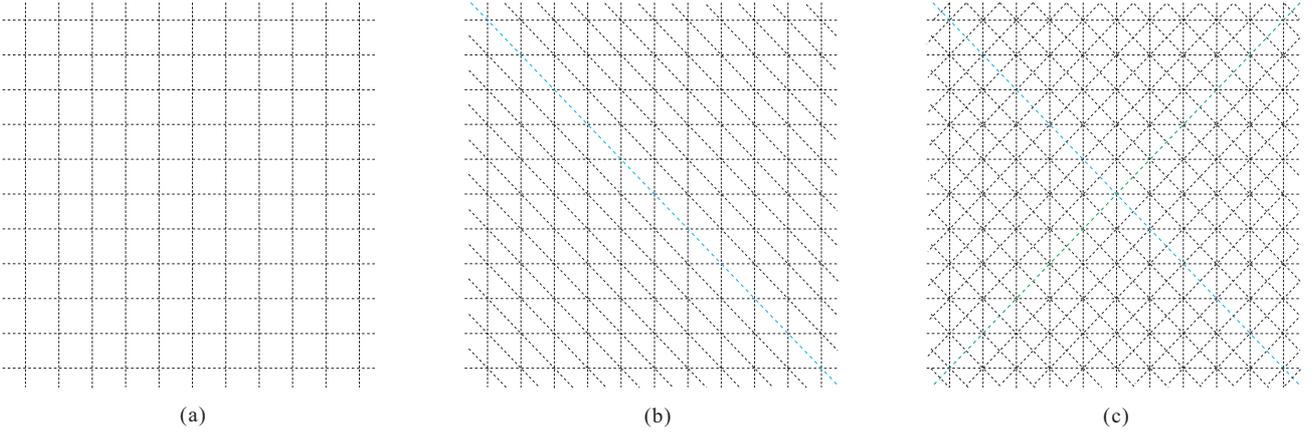}
\caption{A fragment of infinite graph (a) $G^\infty$, (b) $T^\infty$, and (c) $S^\infty$.} \label{Fig_FragmentOfGridRelated}
\end{center}
\end{figure}

\begin{figure}[t]
\begin{center}
\includegraphics[width=0.95\textwidth]{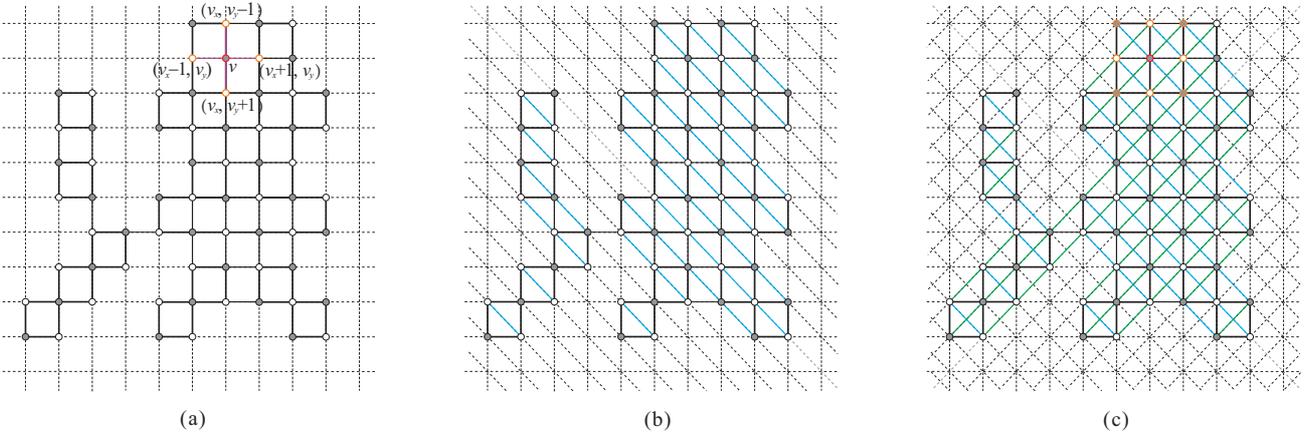}
\caption{(a) A grid graph, (b) a triangular grid graph, and (c) a supergrid graph, where circles represent the vertices and solid lines indicate the edges in the graphs.} \label{Fig_ExampleOfGridRelated}
\end{center}
\end{figure}

The Hamiltonian cycle problem on supergrid graphs can be applied to control the stitching trace of a computerized sewing machine as follows \cite{Hung15}. Consider a computerized sewing machine given an image. The computerized sewing software is used to compute the sewing traces of a computerized sewing machine. A computerized sewing software may include two parts. The first part is to do image processing for the input image, eg., reduce order of colors and image thinning. It then produces some sets of lattices in which every set of lattices represents a color in the input image for sewing. The second part is given by a set of lattices and then computes a cycle (path) to visit the lattices of the set such that each lattice is visited exactly once. Finally, the software transmits the stitching trace of the computed cycle (path) to the computerized sewing machine, and the machine then performs the sewing work along the trace on the object, e.g. clothes. For example, given an image in Fig. \ref{Fig_ComputerizedSewing}(a), the software first analyzes the image and then produces ten colors of regions in which each region is filled with the same color and may be formed by some blocks, as shown in Fig. \ref{Fig_ComputerizedSewing}(b). It then produces ten sets of lattices in which every set of lattices represents a region, where each region is filled by a sewing trace with the same color and it may be partitioned into many non-contiguous blocks. Fig. \ref{Fig_ComputerizedSewing}(c) shows a set of lattices for one region of color, and the software then computes a sewing trace for the set of lattices, as depicted in Fig. \ref{Fig_ComputerizedSewing}(d). Since each stitch position of a sewing machine can be moved to its eight neighbor positions (left, right, up, down, up-left, up-right, down-left, and down-right), one set of adjacent lattices forms a supergrid graph which may be disconnected. Note that each lattice will be represented by a vertex of a supergrid graph, each region may be separated by many blocks in which each block represented a connected supergrid graph. The desired sewing trace of each set of adjacent lattices is the Hamiltonian cycle (path) of the corresponding connected supergrid graph when it is Hamiltonian. Note that if the corresponding supergrid graph is not Hamiltonian, then the sewing trace contains more than one paths and these paths must be concatenated by jump lines. After computing the sewing traces of all regions of colors, the software then transmits the computed stitching trace to the computerized sewing machine. Fig. \ref{Fig_ComputerizedSewing}(e) depicts the final sewing result for the image in Fig. \ref{Fig_ComputerizedSewing}(a). In addition, the structure of supergrid graphs can be used to design the topology of networks, and the network diameter of the proposed supergrid graphs is smaller than that of grid graphs.

\begin{figure}[thp]
\begin{center}
\includegraphics[width=0.9\textwidth]{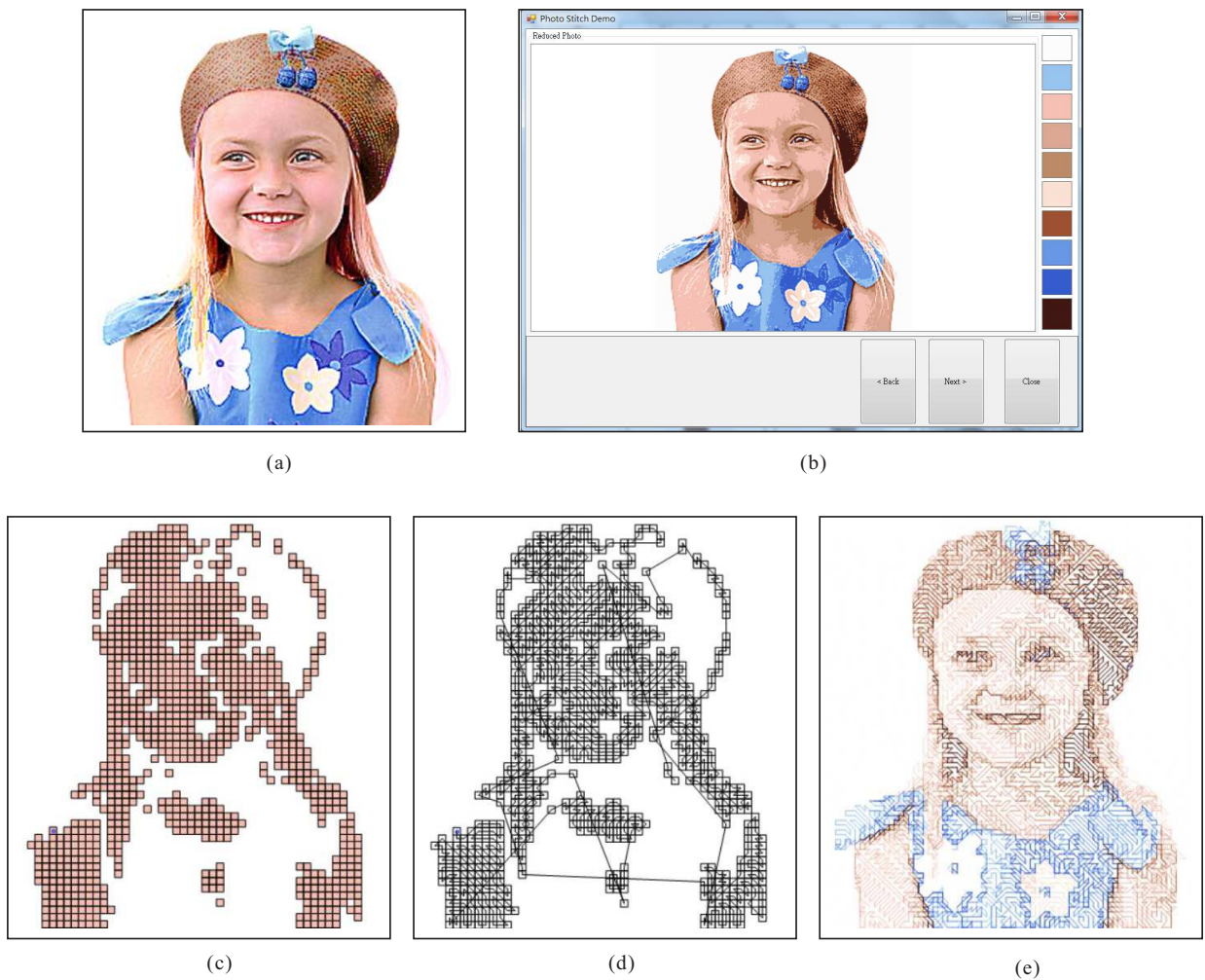}
\caption{(a) An input image for the computerized sewing software, (b) ten colors of regions produced by image processing, (c) a set of lattices for one region of color, (d) a possible sewing trace for the set of lattices in (c), and (e) an overview after computing sewing traces and cutting all jump lines of all regions of colors.} \label{Fig_ComputerizedSewing}
\end{center}
\end{figure}

In this paper, we will study the Hamiltonian property of a subclass of supergrid graphs, called \textit{linear-convex supergrid graphs}. A supergrid graph $G$ is \textit{linearly convex} if, for every line $l$ which contains an edge of $S^\infty$, the intersection of $l$ and $G$ is either a line segment (a path in $G$), or a point (a vertex in $G$), or empty. Linear-convex triangular grid graphs are defined similarly. In general, a linear-convex supergrid graph is always a linear-convex triangular grid graph with the same vertex set, but the reverse is not true. For example, a linear-convex triangular grid graph $T$ is shown in Fig. \ref{Fig_LinConvexTriSuper}(a), but its corresponding supergrid graph with vertex set $V(T)$ is not linear convex, as shown in Fig. \ref{Fig_LinConvexTriSuper}(b). In \cite{Gordon08, Reay00}, the authors showed that any 2-connected, linear-convex triangular grid graph with the exception of the Star of David graph always contains a Hamiltonian cycle. However, the result can not be applied to 2-connected, linear-convex supergrid graphs. In this paper, we first show that any 2-connected, linear-convex supergrid graph is locally connected. We then prove that any 2-connected, linear-convex supergrid graph contains a Hamiltonian cycle.

\begin{figure}[thp]
\begin{center}
\includegraphics[scale=0.9]{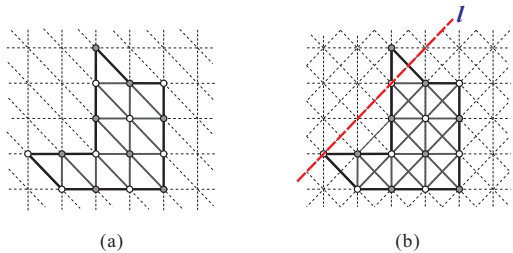}
\caption{(a) A linear-convex triangular grid graph $T$, and (b) a supergrid graph $S$ with vertex set $V(T)$ which is not linearly convex, where the intersection of line $l$ and graph $S$ contains two isolated vertices.} \label{Fig_LinConvexTriSuper}
\end{center}
\end{figure}

Related works of investigation are summarized as follows. Itai et al. \cite{Itai82} showed that the Hamiltonian cycle and path problems for grid graphs are NP-complete. They also gave the necessary and sufficient conditions for a rectangular grid graph having a Hamiltonian path between two given vertices. Zamfirescu et al. \cite{Zamfirescu92} gave the sufficient conditions for a grid graph having a Hamiltonian cycle, and proved that all grid graphs of positive width have Hamiltonian line graphs. Later, Chen et al. \cite{Chen02} improved the Hamiltonian path algorithm of \cite{Itai82} on rectangular grid graphs and presented a parallel algorithm for the Hamiltonian path problem with two given endpoints in rectangular grid graph (mesh). Also there is a polynomial-time algorithm for finding Hamiltonian cycles in solid grid graphs \cite{Lenhart97}. In \cite{Salman05}, Salman introduced alphabet grid graphs and determined classes of alphabet grid graphs which contain Hamiltonian cycles. Recently, Keshavarz-Kohjerdi and Bagheri \cite{Keshavarz12a} gave the necessary and sufficient conditions for the existence of Hamiltonian paths in alphabet grid graphs, and presented linear-time algorithms for finding Hamiltonian paths with two given endpoints in these graphs. Recently, Keshavarz-Kohjerdi et al. \cite{Keshavarz12b} presented a linear-time algorithm for computing the longest path between any two given vertices in rectangular grid graphs. Reay and Zamfirescu \cite{Reay00} proved that all 2-connected, linear-convex triangular grid graphs except one special case contain Hamiltonian cycles. The Hamiltonian cycle (path) for triangular grid graphs has been shown to be NP-complete \cite{Gordon08}. They also proved that all connected, locally connected triangular grid graphs (with one exception as in \cite{Reay00}) contain Hamiltonian cycles. The Hamiltonian cycle problem on hexagonal grid graphs has been shown to NP-complete \cite{Islam07}. Recently, we have proved the Hamiltonian cycle (path) problem for supergrid graphs to be NP-complete \cite{Hung15}. For more works of investigation on grid and triangular grid graphs, we refer readers to \cite{Bouchemakh14, Dettlaff14, Gravier02, Hu12, Marx04, Menke97, Muthammai13, Zhang11}.

The rest of the paper is organized as follows. In Section \ref{Preliminaries}, we give some notations and observations. Section \ref{Local-Connectivity} shows that any 2-connected, linear-convex supergrid graph is locally connected. In Section \ref{HC-LinearConvexSupergrid}, we show that any 2-connected, linear-convex supergrid graph contains a Hamiltonian cycle. Finally, we make some concluding remarks in Section \ref{Conclusion}.

%%% ----------------------------------------------------------------------
\section{Preliminaries}\label{Preliminaries}
%%% ----------------------------------------------------------------------

In this section, we will introduce some technologies and symbols. Some useful observations are also given. For graph-theoretic terminology not defined in this paper, the reader is referred to \cite{Bondy08}. Let $G = (V, E)$ be a graph with vertex set $V(G)$ and edge set $E(G)$. Let $S$ be a subset of vertices in $G$, and let $u, v$ be two vertices in $G$. We write $G[S]$ for the subgraph of $G$ \textit{induced} by $S$, $G-S$ for the subgraph $G[V-S]$, i.e., the subgraph induced by $V-S$. In general, we write $G-v$ instead of $G-\{v\}$. If $(u, v)$ is an edge in $G$, we say that $u$ is \textit{adjacent} to $v$ and $u, v$ are \textit{incident} to edge $(u, v)$. The notation $u\thicksim v$ (resp., $u \nsim v$) means that vertices $u$ and $v$ are adjacent (resp., non-adjacent). A \textit{neighbor} of $v$ in $G$ is any vertex that is adjacent to $v$. We use $N_G(v)$ to denote the set of neighbors of $v$ in $G$ and remove the subscript `$G$' of $N_G(v)$ from the notation if it has no ambiguity. The degree of $v$ is the number of vertices adjacent to vertex $v$, and denoted by $deg(v) = |N(v)|$. The \textit{distance} between $u$ and $v$ is the length of the shortest path between those two vertices.

A path $P$ of length $|P|-1$ in $G$, denoted by $v_1\rightarrow v_2\rightarrow \cdots \rightarrow v_{|P|-1} \rightarrow v_{|P|}$, is a sequence $(v_1, v_2, \cdots, v_{|P|-1}, v_{|P|})$ of vertices such that $(v_i,v_{i+1})\in E(G)$ and $v_i\neq v_{i+1}$ for $1 \leqslant i < |P|$. The first and last vertices visited by $P$ are called the \textit{path-start} and \textit{path-end} of $P$, denoted by $start(P)$ and $end(P)$, respectively. We will use $v_i \in P$ to denote ``$P$ visits $v_i$''.  A path from vertex $v_1$ to another vertex $v_k$ is denoted by $(v_1, v_k)$-path. In addition, we use $P$ to refer to the set of vertices visited by path $P$ if it is understood without ambiguity. On the other hand, a path is called the \textit{reversed path}, denoted by $\textrm{rev}(P)$, of path $P$ if it visits the vertices of $P$ from $end(P)$ to $start(P)$ sequentially; that is, the reversed path $\textrm{rev}(P)$ of $P=v_1\rightarrow v_2\rightarrow \cdots \rightarrow v_{|P|-1} \rightarrow v_{|P|}$ is $v_{|P|}\rightarrow v_{|P|-1}\rightarrow \cdots\rightarrow v_2 \rightarrow v_1$. A cycle is a path $C$ with $|V(C)| \geqslant 3$ and $start(C) \thicksim end(C)$. As usually, we use $P_k$ and $C_k$ to denote the path and cycle of $k$ vertices, respectively.

A cycle $C$ in a graph $G$ is \textit{extendable} if there exists a cycle $C'$ in $G$ (called the \textit{extension} of $C$) such that $V(C)\subset V(C')$ and $|V(C')|=|V(C)|+1$. If such a cycle $C'$ exists, we say that $C$ can be extended to $C'$. A graph $G$ is called \textit{cycle extendable} if $G$ has at least one cycle and if every non-hamiltonian cycle in $G$ is extendable. The cycle extendable graphs form a subclass of Hamiltonian graphs \cite{Hendry90}. Thus, every cycle extendable graph is Hamiltonian.

We next give some observations on the relations among cycle, path, and vertex in the following. Let $C$ be a cycle of a graph and let vertex $x\notin C$. If there exists an edge $u, v$ in $C$ such that $u\thicksim x$ and $v\thicksim x$, then $C\rightarrow x$ is a cycle of the graph, where $start(C), end(C)\in\{u, v\}$. Thus the following proposition holds true.

\begin{pro}\label{CycleVertexMerge_Obs}
Let $C$ be a cycle of a graph $G$ and let $x$ be a vertex in $G-V(C)$. If there exists an edge $(u, v)$ in $C$ such that $u\thicksim x$ and $v\thicksim x$, then $C$ and $x$ can be merged into a cycle of $G$.
\end{pro}

The above observation can be extended to a path $P$ as the following proposition.

\begin{pro}\label{CyclePathMerge_Obs}
Let $C$ and $P$ be a cycle and a path, respectively, of a graph $G$ such that $V(C)\cap V(P)=\emptyset$. If there exists an edge $(u, v)$ in $C$ such that $u\thicksim start(P)$ and $v\thicksim end(P)$, then $C$ and $P$ can be concatenated into a cycle of $G$.
\end{pro}

Let $C_1$ and $C_2$ be two vertex-disjoint cycles of a graph $G$. If there exist two edges $(u_1, v_1)\in C_1$ and $(u_2, v_2)\in C_2$ such that $u_1\thicksim u_2$ and $v_1\thicksim v_2$, then $C_1\rightarrow C_2$ is a cycle of $G$, where $end(C_1)=v_1$ and $start(C_2)=v_2$. Thus we have the following proposition.

\begin{pro}\label{CycleCycleMerge_Obs}
Let $C_1$ and $C_2$ be two vertex-disjoint cycles of a graph $G$. If there exist two edges $(u_1, v_1)\in C_1$ and $(u_2, v_2)\in C_2$ such that $u_1\thicksim u_2$ and $v_1\thicksim v_2$, then $C_1$ and $C_2$ can be concatenated into a cycle of $G$.
\end{pro}

Let $C_1$ and $C_2$ be two cycles of a graph $G$ such that $V(C_1)\cap V(C_2)=\{v\}$. If there exist two edges $(u, v)\in C_1$ and $(w, v)\in C_2$ such that $u\thicksim w$, then $C_1\rightarrow C_2$ forms a cycle of $G$, where $end(C_1)=u$ and $start(C_2)=w$. Hence, the following observation is true.

\begin{pro}\label{CycleCycleVertexMerge_Obs}
Let $C_1$ and $C_2$ be two cycles of a graph $G$ such that $V(C_1)\cap V(C_2)=\{v\}$. If there exist two edges $(u, v)\in C_1$ and $(w, v)\in C_2$ such that $u\thicksim w$, then $C_1$ and $C_2$ can be concatenated into a cycle of $G$.
\end{pro}

Let $S^\infty$ be the infinite graph whose vertex set consists of all points of the plane with integer coordinates and in which two vertices are adjacent if and only if the difference of their $x$ or $y$ coordinate is not larger than 1. A \textit{supergrid graph} is a finite, vertex-induced subgraph of $S^\infty$. For a vertex $v$ in a supergrid graph, let $v_x$ and $v_y$ denote $x$ and $y$ coordinates of $v$, respectively. We color vertex $v$ to be \textit{white} if $v_x+v_y\equiv 0$ (mod 2); otherwise, $v$ is colored to be \textit{black}. Then there are eight possible neighbors of vertex $v$ including four white vertices and four black vertices. For simplicity, we denote the possible neighbors of $v=(v_x, v_y)$ as $UL(v)=(v_x-1, v_y-1)$ (up-left), $U(v)=(v_x, v_y-1)$ (up), $UR(v)=(v_x+1, v_y-1)$ (up-right), $L(v)=(v_x-1, v_y)$ (left), $R(v)=(v_x+1, v_y)$ (right), $DL(v)=(v_x-1, v_y+1)$ (down-left), $D(v)=(v_x, v_y+1)$ (down), and $DR(v)=(v_x+1, v_y+1)$ (down-right). Fig. \ref{Fig_Neighbors} depicts the possible neighbors of vertex $v$ in a supergrid graph. For example, $u=UL(v)$ means that vertex $u$ is the up-left neighbor of vertex $v$.

\begin{figure}[t]
\begin{center}
\includegraphics[scale=1.0]{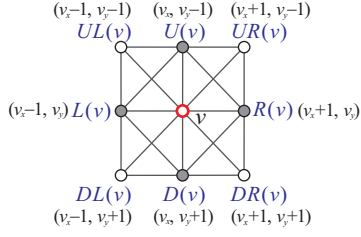}
\caption{The possible neighbors of vertex $v=(v_x, v_y)$ in a supergrid graph.} \label{Fig_Neighbors}
\end{center}
\end{figure}

A supergrid graph $G$ is \textit{linearly convex} if, for every line $l$ which contains an edge of $S^\infty$, the intersection of $l$ and $G$ is either a line segment (a path in $G$), or a point (a vertex in $G$), or empty. For example, Fig. \ref{Fig_LinConvexSuper}(a) is a linear-convex supergrid graph but Fig. \ref{Fig_LinConvexSuper}(b) is not a linear-convex supergrid graph since there exists at least one line $l$ in $S^\infty$ such that the intersection of $l$ and $G$ includes two isolated vertices.

\begin{figure}[t]
\begin{center}
\includegraphics[scale=0.9]{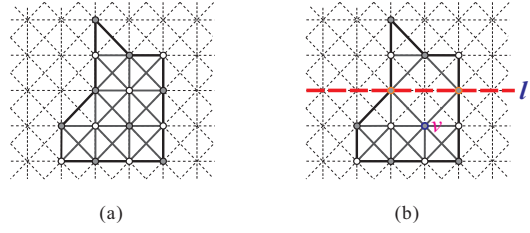}
\caption{(a) A linear-convex supergrid graph, and (b) a supergrid graph which is not linearly convex, where solid lines indicate edges of graphs and the bold dashed line $l$ indicates that the intersection of $l$ and the graph in (b) contains two isolated vertices.} \label{Fig_LinConvexSuper}
\end{center}
\end{figure}

Let $v$ be a vertex in a linear-convex supergrid graph $G$. We observe that if both $UL(v)$ and $UR(v)$ are vertices in $G$, then $U(v)$ is also a vertex in $G$; otherwise, we can find a parallelism line $l$ passing through $UL(v)$, $UR(v)$ but not passing through $U(v)$, and the intersection of $l$ and $G$ contains at least two paths of $G$. For example, Fig. \ref{Fig_LinConvexSuper}(b) depicts such a situation. Thus, we have the following proposition.

\begin{pro}\label{ImmediateVertexInclude}
Let $v$ be a vertex of a linear-convex supergrid graph $G$. If $UL(v)$ and $UR(v)$ (resp., $UL(v)$ and $DL(v)$, $UR(v)$ and $DR(v)$, $DL(v)$ and $DR(v)$) are in $G$, then $U(v)$ (resp., $L(v)$, $R(v)$,$D(v)$ ) is in $G$.
\end{pro}

%%% ----------------------------------------------------------------------
\section{Local connectivity of linear-convex supergrid graphs}\label{Local-Connectivity}
%%% ----------------------------------------------------------------------

In this section, we will establish an interrelation between 2-connected, linear-convex supergrid graphs and locally connected supergrid graphs. Recall that a graph $G$ is $2$-connected if there are two vertex-disjoint paths between every pair of vertices in $G$, and $G$ is locally connected if for every vertex $u$ in $G$, the induced subgraph $G[N(u)]$ of $G$ is connected. A locally connected supergrid graph is not necessary to be linearly convex. For example, Fig. \ref{Fig_LocallyConnect-NotLinear} shows a locally connected supergrid graph, which is not linearly convex: the intersection of the graph and the bold dashed line $l$ is the union of a point (the vertex $u$ of the graph) and a line segment (the edge $(v, w)$ of the graph). In the following theorem, we will prove that any 2-connected, linear-convex supergrid graph is locally connected.

\begin{figure}[t]
\begin{center}
\includegraphics[scale=0.95]{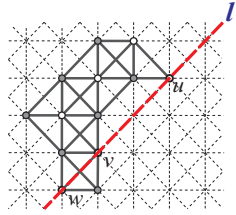}
\caption{A locally connected, but not linear-convex supergrid graph, where solid lines indicate edges of graphs and the bold dashed line $l$ contains edges of $S^\infty$.} \label{Fig_LocallyConnect-NotLinear}
\end{center}
\end{figure}

\begin{thm}\label{LocalConnectivity}
Let $G$ be a $2$-connected, linear-convex supergrid graph. Then, $G$ is locally connected.
\end{thm}
\begin{proof}
We prove this theorem by contradiction. Assume by contradiction that $G$ contains a vertex $u$ such that $u$ is not locally connected. Then, $G[N(u)]$ is not connected. The $2$-connectedness of $G$ implies that $deg(u)\geqslant 2$. On the other hand, if $deg(u)=7$ or $8$, then $G[N(u)]$ is connected and contains a path $P_7$ or $C_8$ as subgraph. Thus, $deg(u)\leqslant 6$. Then, $2\leqslant deg(u)\leqslant 6$. Consider the following five possible cases for the degree of $u$.

\textit{Case} 1: $deg(u)=2$. Let $N(u)=\{v, w\}$. Then, $v\nsim w$. By symmetry, we need only consider the subcases, as shown in Fig. \ref{Fig_LocalConnectivity-deg_2}. Consider the subcase in Fig. \ref{Fig_LocalConnectivity-deg_2}(a). Then, $v=UL(u)$ and $w=UR(u)$. By Proposition \ref{ImmediateVertexInclude}, $U(u)$ is in $G$ and it contradicts that $deg(u)=2$. On the other hand, consider the subcase in Fig. \ref{Fig_LocalConnectivity-deg_2}(b). Then, $v=UL(u)$ and $w=R(u)$. Since $G$ is $2$-connected, there exists a $(v, w)$-path $P$ in $G$ with internal vertices different from $u$. Let $l$ be a line which contains the edge $(U(u), u)$ of $S^\infty$. Then, the intersection of $l$ and $G$ contains vertex $u$ as an isolated vertex (since $U(u)$ and $D(u)$ are not in $G$) and at least one vertex of $P$. This contradicts that $G$ is linearly convex. The other subcases can be proved by the same arguments. Fig. \ref{Fig_LocalConnectivity-deg_2}(c)--(d) also depicts the lines $l$ of $S^\infty$ used to obtain a contradiction.

\begin{figure}[t]
\begin{center}
\includegraphics[scale=1.0]{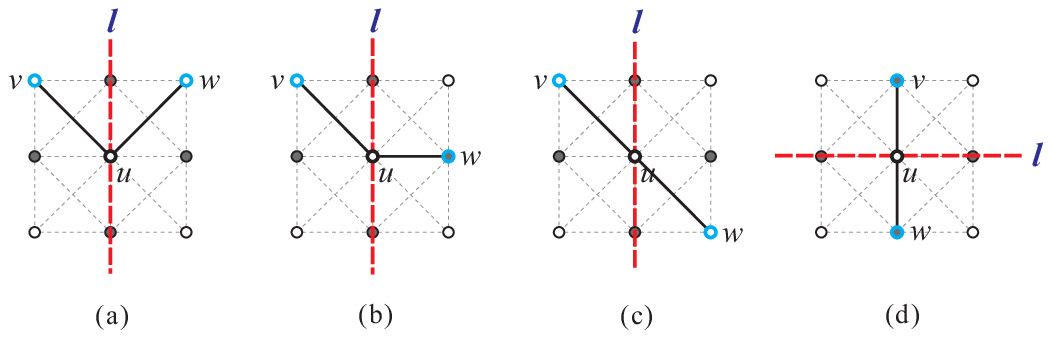}
\caption{The possible subcases for $deg(u)=2$ in the proof of Theorem \ref{LocalConnectivity}, where the lines $l$ of $S^\infty$ used to obtain a contradiction are drawn by bold dashed lines.} \label{Fig_LocalConnectivity-deg_2}
\end{center}
\end{figure}

\textit{Case} 2: $deg(u)=3$. Let $N(u)=\{v, w, z\}$. In this case, the possible related locations among $u, v, w, z$ are depicted in Fig. \ref{Fig_LocalConnectivity-deg_3}. Consider the subcase of Fig. \ref{Fig_LocalConnectivity-deg_3}(a). Then, $v=UL(u)$, $w=U(u)$, and $z=DR(u)$. Since $G$ is $2$-connected, there exists a $(v, z)$-path $P$ in $G$ with internal vertices different from $u$. Let $l$ be a line which contains the edge $(L(u), u)$ of $S^\infty$. Then, the intersection of $l$ and $G$ contains vertex $u$ as an isolated vertex (since $L(u)$ and $R(u)$ are not in $G$) and at least one vertex of $P$. This contradicts that $G$ is linearly convex. By the same arguments, we can find a line $l$ of $S^\infty$ to get a contradiction for the other cases. Fig. \ref{Fig_LocalConnectivity-deg_3}(b)--(f) also depict these lines.

\begin{figure}[tp]
\begin{center}
\includegraphics[scale=1.0]{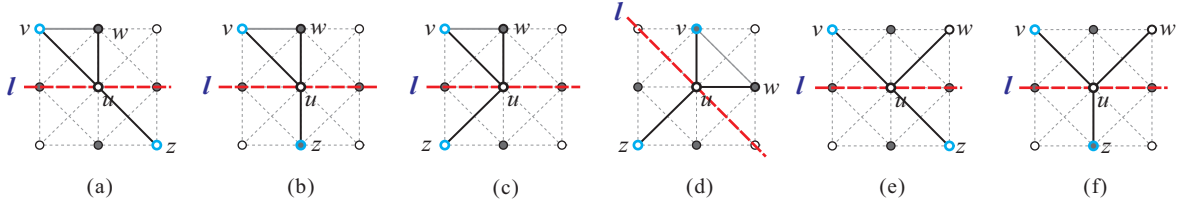}
\caption{The possible subcases for $deg(u)=3$ in the proof of Theorem \ref{LocalConnectivity}, where the lines $l$ of $S^\infty$ used to obtain a contradiction are drawn by bold dashed lines.} \label{Fig_LocalConnectivity-deg_3}
\end{center}
\end{figure}

\textit{Case} 3: $deg(u)=4$. Let $N(u)=\{v, w, z, x\}$. By symmetry, the possible related locations among $u, v, w, z, x$ are shown in Fig. \ref{Fig_LocalConnectivity-deg_4}. For the subcases of Fig. \ref{Fig_LocalConnectivity-deg_4}(a)--(d), we can obtain a contradiction by similar arguments in proving Case 1, where there exists a $(v, x)$-path in $G$ with internal vertices different from $u$. Consider the subcase of  Fig. \ref{Fig_LocalConnectivity-deg_4}(e). Then, $v=UL(u)$, $w=U(u)$, $z=R(u)$, and $x=DL(u)$. Since $G$ is $2$-connected, there exists a $(v, x)$-path $P$ in $G$ with internal vertices different from $u$. Let $l_1$ be a line which contains the edge $(u, z)$ of $S^\infty$, and let $l_2$ be a line which contains the edge $(v, u)$ of $S^\infty$. Clearly, the intersection of $l_1$ and $G$ contains the edge $(u, z)$ but does not contain $L(u)$, and the intersection of $l_2$ and $G$ contains the edge $(v, u)$ but does not contain $DR(u)$. Let $l'_1$ and $l'_2$ be two part rays of $l_1$ and $l_2$, respectively, such that $l'_1$ starts from $u$ and passes $L(u)$, and $l'_2$ starts from $u$ and passes $DR(u)$. Then, the intersection of $l'_1\cup l'_2$ and $G$ contains at least one vertex of $P$. This contradicts that $G$ is linearly convex. By the same arguments, we can arrive at a contradiction for the other subcases of Fig. \ref{Fig_LocalConnectivity-deg_4}(f)--(h) which also indicate lines $l_1$ and $l_2$ in these subcases.

\begin{figure}[tp]
\begin{center}
\includegraphics[scale=1.0]{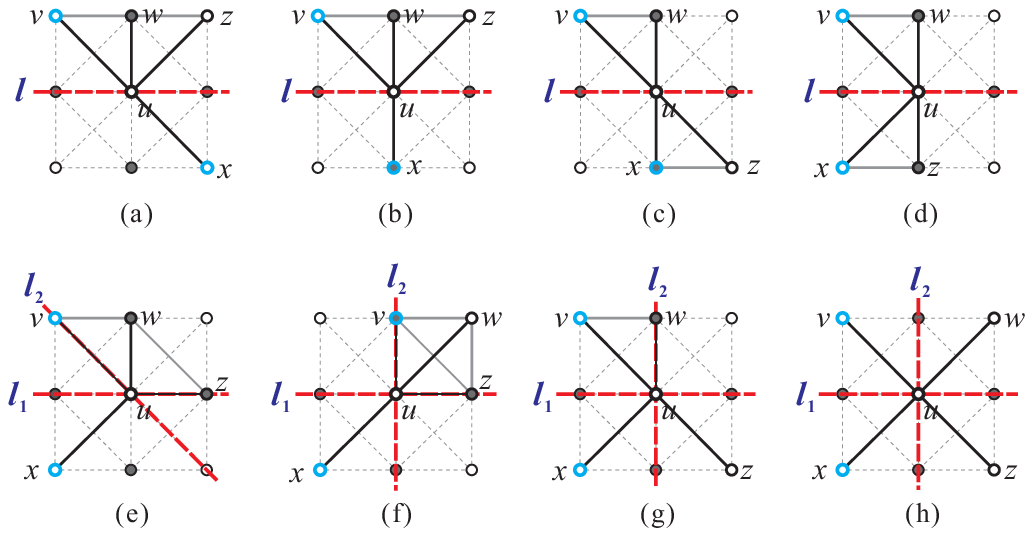}
\caption{The possible subcases for $deg(u)=4$ in the proof of Theorem \ref{LocalConnectivity}, where the lines $l$ of $S^\infty$ used to obtain a contradiction are drawn by bold dashed lines.} \label{Fig_LocalConnectivity-deg_4}
\end{center}
\end{figure}

\textit{Case} 4: $deg(u)=5$. Let $N(u)=\{x_1, x_2, x_3, x_4, x_5\}$. By symmetry, the possible related locations among $u$, $x_1$, $x_2$, $x_3$, $x_4$, $x_5$ are shown in Fig. \ref{Fig_LocalConnectivity-deg_56}(a)--(d). The proofs of subcases in Fig. \ref{Fig_LocalConnectivity-deg_56}(a)--(b) are similar to the proof of Case 1, and the proofs of subcases in Fig. \ref{Fig_LocalConnectivity-deg_56}(c)--(d) are similar to the proof of Case 3 in Fig. \ref{Fig_LocalConnectivity-deg_4}(e), where there exists a $(x_1, x_5)$-path in $G$ with internal vertices different from $u$.

\begin{figure}[tp]
\begin{center}
\includegraphics[scale=1.0]{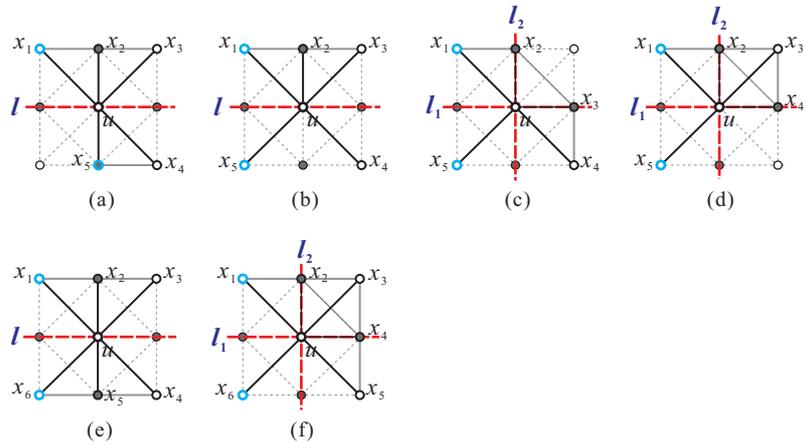}
\caption{The possible subcases for $deg(u)=5$ and $deg(u)=6$ in the proof of Theorem \ref{LocalConnectivity}, where (a)--(d) are subcases of $deg(u)=5$ and (e)--(f) are subcases of $deg(u)=6$.} \label{Fig_LocalConnectivity-deg_56}
\end{center}
\end{figure}

\textit{Case} 5: $deg(u)=6$. Let $N(u)=\{x_1, x_2, x_3, x_4, x_5, x_6\}$. By symmetry, the possible related locations among $u, x_1, x_2, x_3, x_4, x_5, x_6$ are shown in Fig. \ref{Fig_LocalConnectivity-deg_56}(e)--(f). The proof of subcase in Fig. \ref{Fig_LocalConnectivity-deg_56}(e) is similar to the proof of Case 1. And the proof of subcase in Fig. \ref{Fig_LocalConnectivity-deg_56}(f) is similar to the proof of Case 3 in Fig. \ref{Fig_LocalConnectivity-deg_4}(e). Note that there exists a $(x_1, x_6)$-path in $G$ with internal vertices different from $u$.

It follows from the above cases that a contradiction occurs, and, hence, every vertex of a $2$-connected, linear-convex supergrid graph is locally connected. Thus, a $2$-connected, linear-convex supergrid graphs is locally connected. This completes the proof.
\end{proof}

Let $G$ be a $2$-connected, linear-convex supergrid graph and let $u$ be a vertex in $G$. Since $G$ is $2$-connected, $deg(u)\geqslant 2$. By Theorem \ref{LocalConnectivity}, $G$ is locally connected. Thus, there exist two vertices $v, w$ of $N(u)$ such that $(v, w)$ is an edge of $G$. Then, $u\rightarrow v\rightarrow w$ is a cycle of $G$. Thus, we have the following corollary.

\begin{cor}\label{CyleContaining}
If $G$ is a $2$-connected, linear-convex supergrid graph, then $G$ contains at least one cycle.
\end{cor}

%%% ----------------------------------------------------------------------
\section{The Hamiltonian cycle problem on linear-convex supergrid graphs}\label{HC-LinearConvexSupergrid}
%%% ----------------------------------------------------------------------

In this section, we consider connected graphs on $n\geqslant 3$ vertices. We will prove that every $2$-connected, linear-convex supergrid graph is cycle extendable. The result implies that every $2$-connected, linear-convex supergrid graph contains a Hamiltonian cycle. The main result is shown as follows.

\begin{thm}\label{CycleExtendable}
Let $G$ be a $2$-connected, linear-convex supergrid graph. Then, $G$ is cycle extendable.
\end{thm}
\begin{proof}
By Corollary \ref{CyleContaining}, $G$ contains at least one cycle. Since $G$ is 2-connected, the degree of every vertex in $G$ is at least 2. We will prove this theorem by contradiction. Assume by contradiction that $G$ is not cycle extendable. Then, there exists a non-extendable, non-Hamiltonian cycle $C$ of $G$. Let $C = u_1\rightarrow u_2\rightarrow u_3\rightarrow \cdots \rightarrow u_{k-1}\rightarrow u_k$ with $k<|V(G)|$, where $u_1\thicksim u_k$. Since $G$ is connected, there exists a vertex $x$ not in $C$ which is adjacent to one vertex lying on $C$. Without loss of generality, let $u_1$ be such a vertex on $C$ with $u_1 \thicksim x$. Then, $3\leqslant deg(u_1)\leqslant 8$. By Proposition \ref{CycleVertexMerge_Obs}, if $x\thicksim u_2$ or $x\thicksim u_k$ then $C$ can be extended to cover $x$, a contradiction. Thus, $x\nsim u_2$ and $x\nsim u_k$. Depending on whether $u_2 \thicksim u_k$ and the symmetry, we need only consider the cases, as shown in Fig. \ref{Fig_RelatedLocation_x-u1}, for the related positions among $u_1, u_2, u_k, x$.

\begin{figure}[t]
\begin{center}
\includegraphics[scale=1.0]{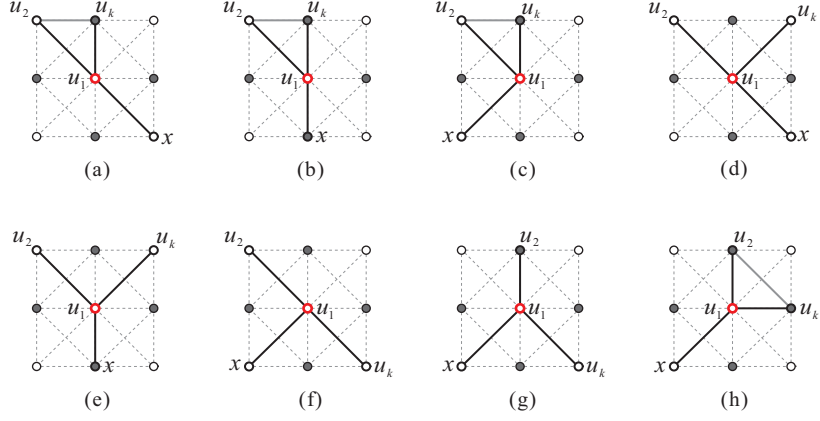}
\caption{The possible relative locations among $u_1, u_2, u_k, x$ when $x\sim u_1$, $x\nsim u_2$, and $x\nsim u_k$.} \label{Fig_RelatedLocation_x-u1}
\end{center}
\end{figure}

Let $Z=\{L(u_1), R(u_1)\}-\{u_2, u_k\}$. We can select set $Z$ to satisfy that $Z\cap V(G)\neq\emptyset$ as follows. Consider the cases of Fig. \ref{Fig_RelatedLocation_x-u1}(a)--(g). Suppose that $Z\cap V(G)=\emptyset$. Since $G$ is $2$-connected, there exists a $(u_2, x)$-path $P$ in $G$ with internal vertices different from $u_1$. Let $l$ be a line which contains edge $(L(u_1), u_1)$ of $S^\infty$. Then, the intersection of $l$ and $G$ contains vertex $u_1$ as an isolated vertex and at least one vertex of $P$. It contradicts that $G$ is linearly convex. Thus, $Z\cap V(G)\neq\emptyset$. On the other hand, consider the case of Fig. \ref{Fig_RelatedLocation_x-u1}(h). By similar arguments in proving Case 3 (Fig. \ref{Fig_LocalConnectivity-deg_4}(f)) of Theorem \ref{LocalConnectivity}, we can verify that $\{L(u_1), D(u_1)\}\cap V(G)\neq\emptyset$. By symmetry, we can let $L(u_1)\in V(G)$. Thus, we can select $Z$ with $Z\cap V(G)\neq\emptyset$. We then  select a vertex $z$ of $Z$ to satisfy the following condition:\\

\noindent(C1) if one of $Z\cap V(G)$ is adjacent to vertex $x$, then let $z$ be the vertex with $z \thicksim x$; otherwise, $z\in Z\cap V(G)$ with $z \nsim x$.\\

The above condition implies that if $z \nsim x$ then none of $Z\cap V(G)$ is adjacent to $x$. By inspecting all cases of Fig. \ref{Fig_RelatedLocation_x-u1}, $z$ is adjacent to at least one of $u_2$ and $u_k$. Since $(u_1, u_2)\in C$, $(u_1, u_k)\in C$, $z\thicksim u_1$, and ($z\thicksim u_2$ or $z\thicksim u_k$), by Proposition \ref{CycleVertexMerge_Obs} $z\in C$. If $z \thicksim u_k$ and $z \nsim u_2$, then we can exchange the positions of $u_2$ and $u_k$ because of the symmetric structure of a cycle. Thus, we may assume that $z \thicksim u_2$ below. We proceed via the following two statements toward a final contradiction.\\
1. If $z \thicksim x$, then $C$ can be extended to cover $x$, i.e., $C$ and $x$ can be merged into a cycle.\\
2. If $z \nsim x$, then $C$ can be extended to cover $x$.

Summarizing the above statements, we conclude that $C$ is cycle extendable. It contradicts that $C$ is a non-extendable, non-Hamiltonian cycle of $G$. This completes the proof of the theorem.
\end{proof}

The proofs of the above two statements are given below as Claims 1--2. We will prove these two claims by constructing a cycle $C'$ extended from $C$ to cover vertex $x$.\\

\noindent \textbf{Claim 1.} If $z \thicksim x$, then $C$ can be extended to cover $x$.\\

\noindent \textbf{Proof.} By inspecting the cases in Fig. \ref{Fig_RelatedLocation_x-u1}, the possible locations among $u_1, u_2, u_k, x$, and $z$ under $z \thicksim x$ are shown in Fig. \ref{Fig_RelatedLocation_x-u1-z-Claim1}. Note that $z \thicksim u_2$, $z \thicksim x$, and we can exchange the locations of $u_2$ and $u_k$ by symmetry. Since $(u_1, u_2)\in C$, $z \thicksim u_1$, and $z \thicksim u_2$, we get $z\in C$ by Proposition \ref{CycleVertexMerge_Obs}. Let $z = u_j$. Then, $2 < j < k$. If $x \thicksim u_{j-1}$ or $x \thicksim u_{j+1}$, then by Proposition \ref{CycleVertexMerge_Obs}, $C$ can be extended to cover $x$ since $x \thicksim z$. In the following, assume that $x \nsim u_{j-1}$ and $x \nsim u_{j+1}$. Consider that $u_{j-1} \thicksim u_{j+1}$. Let $P_1 = u_2\rightarrow u_3\rightarrow \cdots \rightarrow u_{j-1}$ and let $P_2 = u_{j+1}\rightarrow u_{j+2}\rightarrow \cdots \rightarrow u_k$. Then, $u_1\rightarrow x\rightarrow u_j (=z) \rightarrow P_1 \rightarrow P_2$ is a cycle of $G[V(C)\cup \{x\}]$, a contradiction. In the following, assume that $u_{j-1} \nsim u_{j+1}$. Then, the following condition is satisfied:\\

\noindent(C2) $z(=u_j) \thicksim u_2$, $z \thicksim x$, $x \nsim u_{j-1}$, $x \nsim u_{j+1}$, and $u_{j-1} \nsim u_{j+1}$.\\

\begin{figure}[t]
\begin{center}
\includegraphics[scale=1.0]{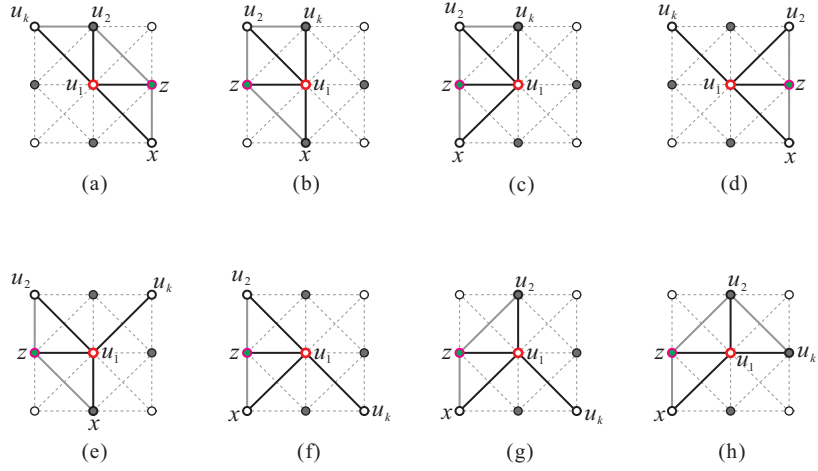}
\caption{The possible relative locations among $u_1, u_2, u_k, x$, and $z$ under that $z\thicksim x$ and $z\thicksim u_2$, where $z\in\{L(u_1), R(u_1)\}-\{u_2, u_k\}$.} \label{Fig_RelatedLocation_x-u1-z-Claim1}
\end{center}
\end{figure}

We first consider the case of $j=3$. Then, $u_2 = u_{j-1}$ and hence $(u_2, u_3=z)$ is an edge in $C$. Suppose that $u_2 \thicksim u_k$. Let $C_1 = u_1\rightarrow x \rightarrow u_3(=z)$ and let $C_2 = u_2\rightarrow u_3\rightarrow \cdots \rightarrow u_{k-1}\rightarrow u_k$. Then, $V(C_1)\cap V(C_2) = \{u_3=z\}$. Since $u_1 \thicksim u_2$, $C_1$ and $C_2$ can be merged into a cycle by Proposition \ref{CycleCycleVertexMerge_Obs}. On the other hand, suppose that $u_2 \nsim u_k$. By Proposition \ref{CycleVertexMerge_Obs}, $x \nsim u_4$ since $x \sim u_3 = z$ and $(u_3, u_4)\in C$. Then, $u_4\in N(z)$, $u_2 \nsim u_4$ ($u_{j-1} \nsim u_{j+1}$), $u_2 \nsim u_k$, and $x \nsim u_4$. That is, the possible cases must satisfy that $N(z)-(N(x)\cup N(u_2))\neq\emptyset$. By examining every case in Fig. \ref{Fig_RelatedLocation_x-u1-z-Claim1}, the possible cases (Fig. \ref{Fig_RelatedLocation_x-u1-z-Claim1}(e) and (g)) are depicted in Fig. \ref{Fig_Claim1_j3-Cases} under that condition (C2) is satisfied, $u_3=z$, and $u_2\nsim u_k$. Consider the case of Fig. \ref{Fig_Claim1_j3-Cases}(a). Then, $N(z)-(N(x)\cup N(u_2))=\{DL(z)\}$. Thus, $u_4 = DL(z)$. Since $DL(z)$ and $DR(z)(=x)$ are in $G$, we obtain that $D(z)$ is in $G$ by Proposition \ref{ImmediateVertexInclude}. Since $D(z) \thicksim z(=u_3)$, $D(z) \thicksim DL(z)=u_4$, and $(z, DL(z))\in C$, we get that $D(z)\in C$ by Proposition \ref{CycleVertexMerge_Obs}. Let $y=D(z)=u_t$. Then, $4 < t \leqslant k-1$. Clearly, $x \thicksim z$, $y \thicksim z$, and $y \thicksim x$. If $u_{t-1} \thicksim u_{t+1}$, then $u_1\rightarrow x\rightarrow y(=u_t)\rightarrow z(=u_3)\rightarrow u_4\rightarrow u_5\rightarrow \cdots \rightarrow u_{t-1}\rightarrow u_{t+1}\rightarrow u_{t+2}\rightarrow \cdots \rightarrow u_{k-1}\rightarrow u_k$ forms a cycle containing $V(C)$ and $x$. Suppose that $u_{t-1} \nsim u_{t+1}$. If $u_{t-1} \thicksim x$ or $u_{t+1} \thicksim x$, then $x$ can be merged into $C$ by Proposition \ref{CycleVertexMerge_Obs}. Thus, $u_{t-1} \nsim x$ and $u_{t+1} \nsim x$. Then, $\{u_{t-1}, u_{t+1}\}=\{UL(y), DL(y)\}$. Let $u_{t-1}=UL(y)$ and $u_{t+1}=DL(y)$. Then, $C= u_1\rightarrow u_2\rightarrow u_3(=z)\rightarrow u_4(=DL(z))\rightarrow P_1 \rightarrow u_{t-1}\rightarrow y(=u_t)\rightarrow u_{t+1}\rightarrow P_2$, where $end(P_2)=u_k$. Let $C' = u_1\rightarrow u_2\rightarrow u_{t-1}\rightarrow \textrm{rev}(P_1)\rightarrow u_4 \rightarrow z\rightarrow x\rightarrow y\rightarrow u_{t+1} \rightarrow P_2$. Then, $C'$ is a cycle that contains $V(C)$ and $x$. The construction of such a cycle $C'$ is depicted in Fig. \ref{Fig_Claim1_j3-Construction}(a). When $u_{t-1}=DL(y)$ and $u_{t+1}=UL(y)$, the extended cycle from $C$ and $x$ is constructed, as shown in Fig. \ref{Fig_Claim1_j3-Construction}(b). By the similar construction, we can construct an extended cycle $C'$ for the case of Fig. \ref{Fig_Claim1_j3-Cases}(b), as shown in Fig. \ref{Fig_Claim1_j3-Construction}(c)--(d).

\begin{figure}[t]
\begin{center}
\includegraphics[scale=1.0]{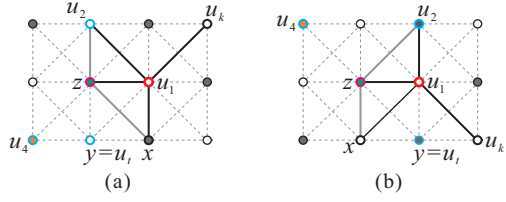}
\caption{The possible cases (Fig. \ref{Fig_RelatedLocation_x-u1-z-Claim1}(e) and (g)) under that condition (C2) is satisfied, $z=u_3$, and $u_2\nsim u_k$.} \label{Fig_Claim1_j3-Cases}
\end{center}
\end{figure}

\begin{figure}[t]
\begin{center}
\includegraphics[scale=1.0]{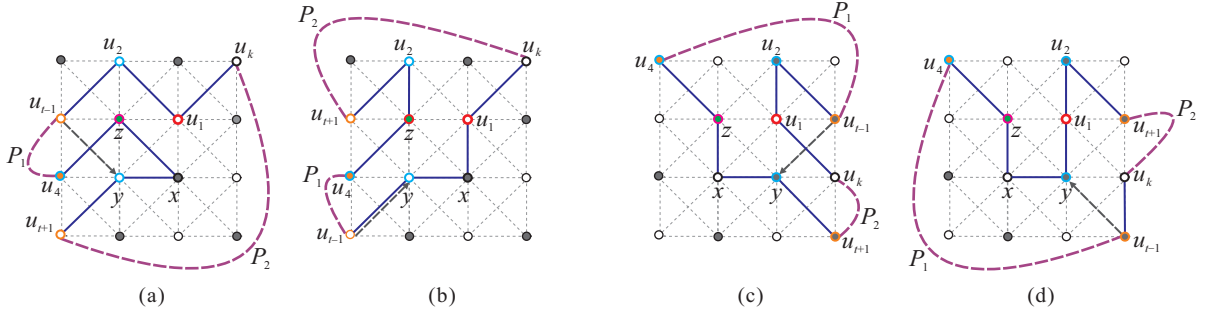}
\caption{(a)--(b) The constructions of extended cycles $C'$ for the case of Fig. \ref{Fig_Claim1_j3-Cases}(a); and (c)--(d) the constructions of extended cycles $C'$ for the case of Fig. \ref{Fig_Claim1_j3-Cases}(b), where solid edges indicate the edges in the cycles, bold dashed lines represent the subpaths $P_1$ and $P_2$ of cycle $C$, and dashed arrow line indicates edge $(u_{t-1}, u_t(=y))$ in $C$.} \label{Fig_Claim1_j3-Construction}
\end{center}
\end{figure}

In the following, we will consider the case of $j>3$. Then, $u_2 \neq u_{j-1}$. Note that $z=u_j$. By condition (C2), $x \nsim u_{j-1}$, $x \nsim u_{j+1}$, and $u_{j-1} \nsim u_{j+1}$. %Then, $G[N(z)-(\{u_1, u_2, u_k\}\cup N(x))]$ is disconnected.
By inspecting all cases in Fig. \ref{Fig_RelatedLocation_x-u1-z-Claim1}, the cases of Fig. \ref{Fig_RelatedLocation_x-u1-z-Claim1}(a)(c)(g)(h) do not satisfy condition (C2). Thus, the possible cases under condition (C2) and $j>3$ are depicted in Fig. \ref{Fig_Claim1_jx3-Cases}. Consider the cases of Fig. \ref{Fig_Claim1_jx3-Cases}(b) and (d). We can see that $\{u_{j-1}, u_{j+1}\}=\{UL(z), UR(z)\}$. Let $C = u_1\rightarrow u_2\rightarrow P_1 \rightarrow z \rightarrow P_2 \rightarrow u_k$, where $end(P_1)=u_{j-1}$ and $start(P_2)=u_{j+1}$. Let $C' = u_1\rightarrow x\rightarrow z\rightarrow \textrm{rev}(P_1)\rightarrow u_2\rightarrow P_2\rightarrow u_k$. Then, $C'$ is a Hamiltonian cycle of $G[V(C)\cup\{x\}]$ and is extended from $C$ by adding $x$. Now, consider the case of Fig. \ref{Fig_Claim1_jx3-Cases}(a). We can see that $\{u_{j-1}, u_{j+1}\} = \{UL(z), DL(z)\}$. Since $UL(z)$, $DL(z)$, and $DR(z)(=x)$ are in $G$, both $L(z)$ and $D(z)$ are in $G$ by Proposition \ref{ImmediateVertexInclude}. Let $y=D(z)$. Then, $y \thicksim x$. Since $y\thicksim z$, $y\thicksim DL(z)$, and $(z, DL(z))\in C$, we get that $y\in C$ by Proposition \ref{CycleVertexMerge_Obs}. Let $y=u_t$. If $x\thicksim u_{t-1}$ or $x\thicksim u_{t+1}$, then a Hamiltonian cycle of $G[V(C)\cup\{x\}]$ can be constructed by Proposition \ref{CycleVertexMerge_Obs}. Thus, assume that $x\nsim u_{t-1}$ and  $x\nsim u_{t+1}$ below. Then, $L(y)(=DL(z))\in \{u_{j-1}, u_{j+1}\}$ and $u_{t-1}, u_{t+1}\in\{UL(y), L(y), DL(y)\}$. Thus, $1\geqslant |\{u_{j-1}, u_{j+1}\} \cap \{u_{t-1}, u_{t+1}\}| \geqslant 0$. There are two cases:

\begin{figure}[t]
\begin{center}
\includegraphics[scale=1.0]{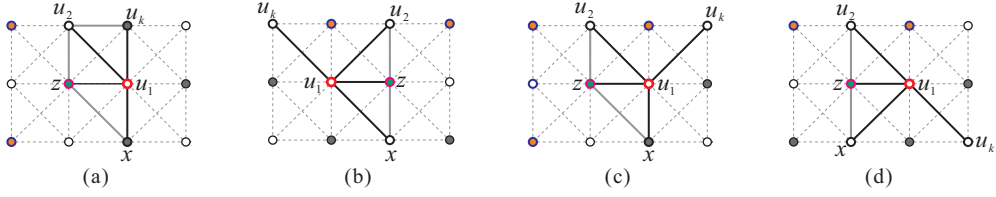}
\caption{The possible cases under condition (C2) and $j>3$.} \label{Fig_Claim1_jx3-Cases}
\end{center}
\end{figure}

\textit{Case} 1: $|\{u_{j-1}, u_{j+1}\} \cap \{u_{t-1}, u_{t+1}\}| = 1$. Then, $\{u_{j-1}, u_{j+1}\} \cap \{u_{t-1}, u_{t+1}\} = \{L(y)\}$. Depending on whether $y$ appears before $z$ in $C$, we consider the following two subcases:

\hspace{0.5cm}\textit{Case} 1.1: $2 \leqslant t\leqslant j-2$. In this subcase, $y(=u_t)$ appears before $z(=u_j)$ in $C$. Then, $u_{t+1}=u_{j-1}=L(y)=DL(z)$. Let $C = u_1\rightarrow u_2\rightarrow P_1\rightarrow y \rightarrow u_{t+1}\rightarrow z \rightarrow P_2\rightarrow u_k$, where $end(P_1)=u_{t-1}$ and $start(P_2)=u_{j+1}=UL(z)$. Let $C' = u_1\rightarrow x\rightarrow y\rightarrow u_{t+1}\rightarrow \textrm{rev}(P_1)\rightarrow u_2\rightarrow u_k\rightarrow \textrm{rev}(P_2)\rightarrow z$. Then, $C'$ is a Hamiltonian cycle of $G[V(C)\cup\{x\}]$. Fig. \ref{Fig_Claim1_jx3-Case1}(a) (resp., Fig. \ref{Fig_Claim1_jx3-Case1}(b)) shows the cycle $C'$ when $u_{t-1}=UL(y)$ (resp., $u_{t-1}=DL(y)$).

\begin{figure}[tp]
\begin{center}
\includegraphics[scale=1.0]{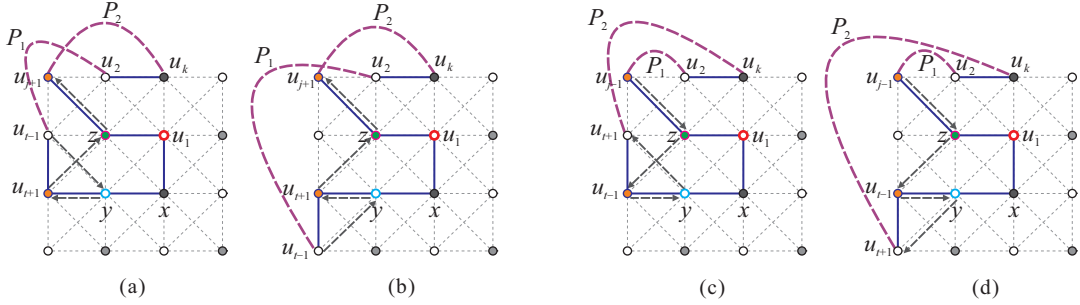}
\caption{The construction of extended cycle $C'$ for (a)--(b) Case 1.1 of Claim 1 under the case of Fig. \ref{Fig_Claim1_jx3-Cases}(a), and (c)--(d) Case 1.2 of Claim 1 under the case of Fig. \ref{Fig_Claim1_jx3-Cases}(a), where solid lines indicate the edges in $C'$, bold dashed lines represent the subpaths of cycle $C$, and dashed arrow lines indicates edges in $C$.} \label{Fig_Claim1_jx3-Case1}
\end{center}
\end{figure}

\hspace{0.5cm}\textit{Case} 1.2: $j+1 \leqslant t\leqslant k-1$. In this subcase, $y(=u_t)$ appears after $z(=u_j)$ in $C$. Then, $u_{j+1}=u_{t-1}=L(y)=DL(z)$. By similar arguments in proving Case 1.1, we can construct an extended cycle $C'$ from $C$ by adding $x$, as shown in Fig. \ref{Fig_Claim1_jx3-Case1}(c)--(d).

\textit{Case} 2: $|\{u_{j-1}, u_{j+1}\} \cap \{u_{t-1}, u_{t+1}\}| = 0$. In this case, $\{u_{j-1}, u_{j+1}\} = \{UL(z), DL(z)\}$ and $\{u_{t-1}, u_{t+1}\} = \{UL(y), DL(y)\}$. There are two subcases:

\hspace{0.5cm}\textit{Case} 2.1: $2 \leqslant t\leqslant j-2$. Since $\{u_{j-1}, u_{j+1}\} \cap \{u_{t-1}, u_{t+1}\} = \emptyset$, $t \neq j-2$. The possible cycles of $C$ are shown in Fig. \ref{Fig_Claim1_jx3-Case21}. Our strategy for constructing a cycle $C'$ combining $C$ with $x$ is sketched as follows. Let $P_1, P_2, P_3$ be three subpaths of $C$ for connecting $u_2, u_{j-1}, u_{j+1}, u_{t-1}, u_{t+1}, u_k$. We first construct a cycle $C_1$ consisting of $P_1, P_2, P_3, z$, and $y$ such that there exists one edge in $C_1$ that is incident to edge $(u_1, x)$. By Proposition \ref{CyclePathMerge_Obs}, $C_1$ and $(u_1, x)$ can be concatenated into a cycle $C'$. Then, $C'$ is the desired cycle constructed from  $C$ to cover $x$. Fig. \ref{Fig_Claim1_jx3-Case21-Construction}(a)--(d) show the constructed cycles $C'$ for Fig. \ref{Fig_Claim1_jx3-Case21}(a)--(d), respectively.

\begin{figure}[t]
\begin{center}
\includegraphics[scale=1.0]{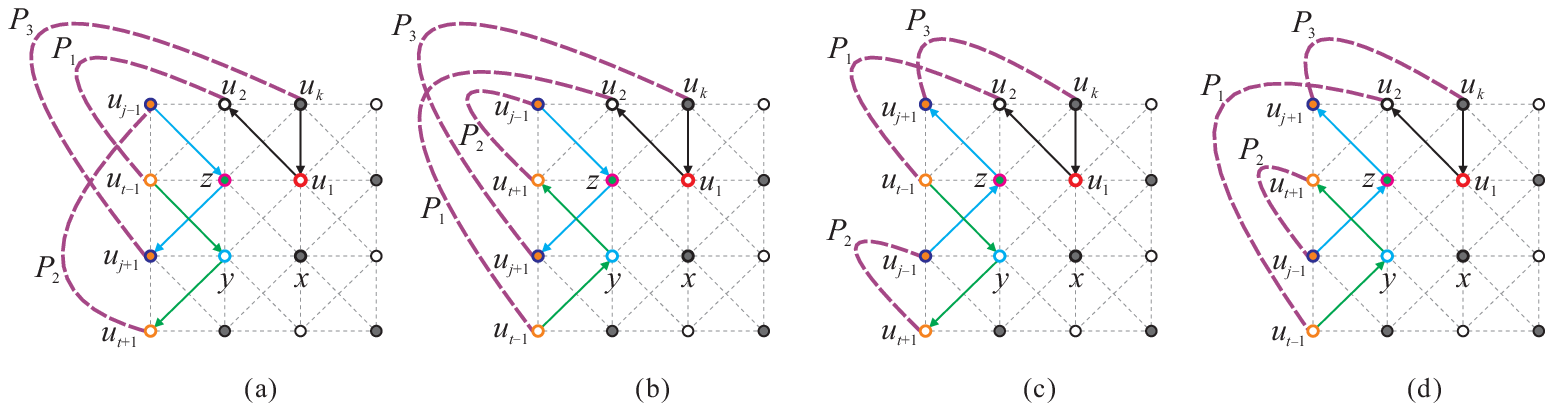}
\caption{The possible cycles $C$ for Case 2.1 of Claim 1 under the case of Fig. \ref{Fig_Claim1_jx3-Cases}(a), where $y$ appears before $z$ in $C$, arrow lines indicate the edges in $C$, and bold dashed lines represent the subpaths of cycle $C$.} \label{Fig_Claim1_jx3-Case21}
\end{center}
\end{figure}

\begin{figure}[tp]
\begin{center}
\includegraphics[scale=1.0]{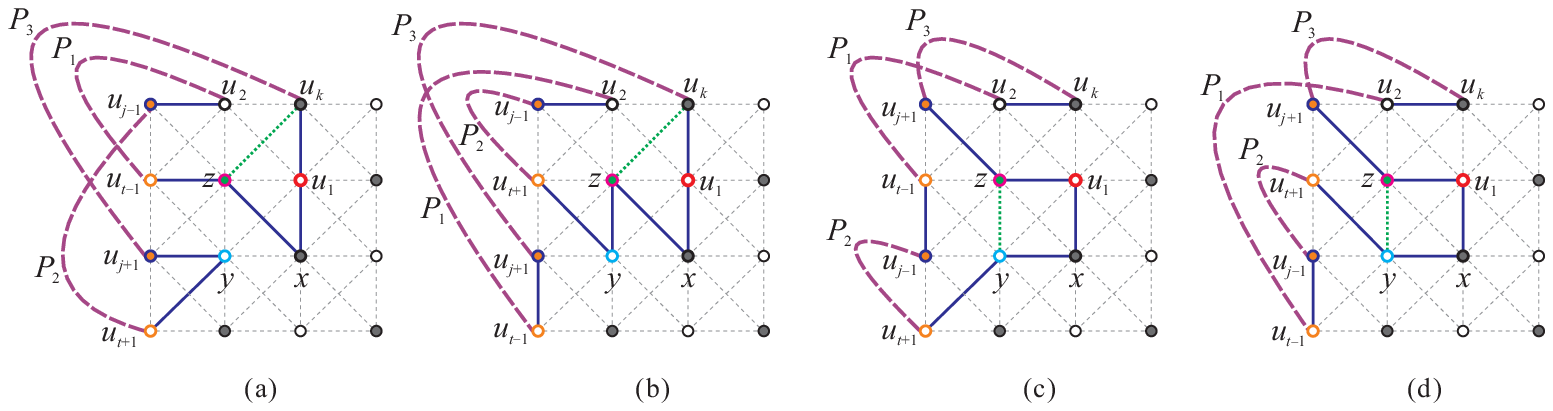}
\caption{The constructed extending cycles for Fig. \ref{Fig_Claim1_jx3-Case21}(a)--(d), respectively, where dashed dark lines represent subpaths of $C$, solid lines indicate the edges in the constructed cycle, and dotted line indicate the edge in the immediate cycle $C_1$ that is incident to edge $(u_1, x)$.} \label{Fig_Claim1_jx3-Case21-Construction}
\end{center}
\end{figure}

\hspace{0.5cm}\textit{Case} 2.2: $j+2 \leqslant t\leqslant k-1$. Since $\{u_{j-1}, u_{j+1}\} \cap \{u_{t-1}, u_{t+1}\} = \emptyset$, $t \neq j+2$. Then, $y$ appears after $z$ in $C$. By similar arguments in proving Case 2.1, we can construct a cycle $C'$ from $V(C)$ and $x$, as shown in Fig. \ref{Fig_Claim1_jx3-Case22-Construction}.

\begin{figure}[tp]
\begin{center}
\includegraphics[scale=1.0]{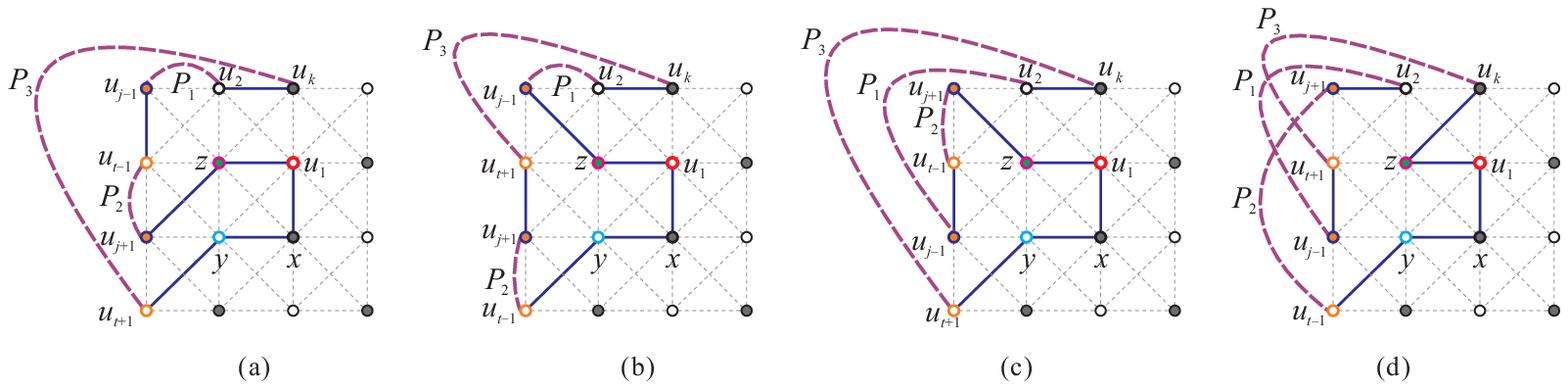}
\caption{The constructed cycles $C'$ for Case 2.2 of Claim 1 under the case of Fig. \ref{Fig_Claim1_jx3-Cases}(a), where $y$ appears after $z$ in $C$, arrow lines indicate the edges in $C$, and bold dashed lines represent the subpaths of cycle $C$.} \label{Fig_Claim1_jx3-Case22-Construction}
\end{center}
\end{figure}

We have proved the cases of Fig. \ref{Fig_Claim1_jx3-Cases}(a), (b), and (d). For the case of Fig. \ref{Fig_Claim1_jx3-Cases}(c), the possible locations of $u_{j-1}$ and $u_{j+1}$ are shown in Fig. \ref{Fig_Claim1_jx3-Case-(c)}. By similar arguments in proving the case of Fig. \ref{Fig_Claim1_jx3-Cases}(b), the extended cycle $C'$ from $C$ and $x$ for the cases of Fig. \ref{Fig_Claim1_jx3-Case-(c)}(a)--(g) can be constructed. Note that Fig. \ref{Fig_Claim1_jx3-Case-(c)}(a)--(e) also show the constructed cycles, where dashed dark lines represent subpaths of $C$ and solid lines indicate the edges in the constructed cycles. In addition, the extended cycle $C'$ for the case of Fig. \ref{Fig_Claim1_jx3-Case-(c)}(h) can be constructed by the similar construction of Fig. \ref{Fig_Claim1_jx3-Cases}(a).

\begin{figure}[tp]
\begin{center}
\includegraphics[scale=1.0]{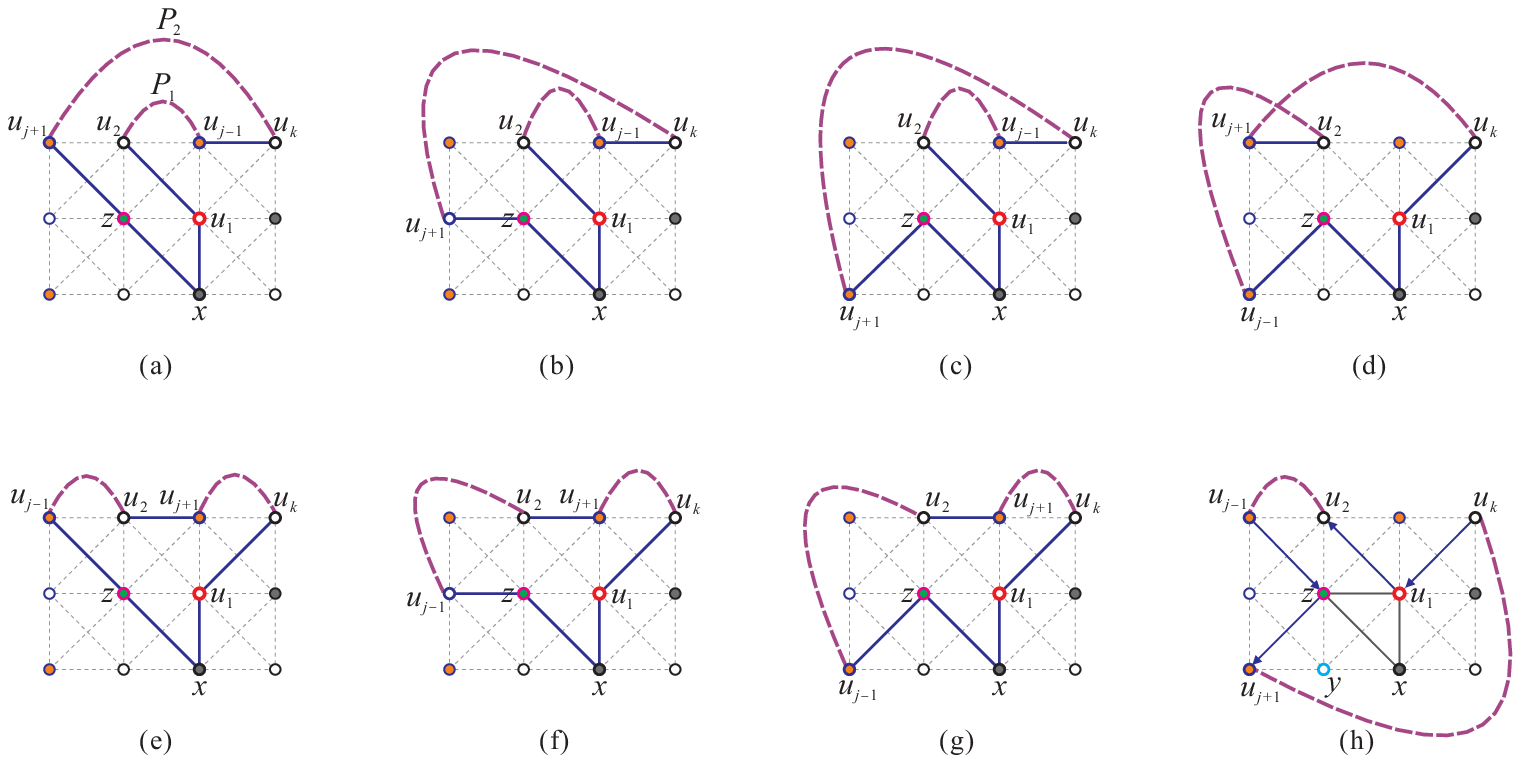}
\caption{The possible locations of $u_{j-1}$ and $u_{j+1}$ for the case of Fig. \ref{Fig_Claim1_jx3-Cases}(c).} \label{Fig_Claim1_jx3-Case-(c)}
\end{center}
\end{figure}

We have considered all cases under $z \thicksim x$ to construct a cycle $C'$ extended from $C$ to cover $x$. This completes the proof of Claim 1.\hfill$\square$\\

\noindent \textbf{Claim 2.} If $z \nsim x$, then $C$ can be extended to cover $x$.\\

\noindent \textbf{Proof.} Recall that $Z=\{L(u_1), R(u_1)\}-\{u_2, u_k\}$, $Z\cap V(G)\neq\emptyset$, and $z\in Z$. We select vertex $z$ to satisfy condition (C1) that if $z \nsim x$ then none of $Z\cap V(G)$ is adjacent to $x$. By inspection on every case of Fig. \ref{Fig_RelatedLocation_x-u1}, the possible relative positions among $z, x, u_1, u_2, u_k$ under $z \nsim x$ are shown in Fig. \ref{Fig_RelatedLocation_Claim2}. For example, in inspecting the case of Fig. \ref{Fig_RelatedLocation_x-u1}(c), by Proposition \ref{ImmediateVertexInclude} $z=L(u_1)\thicksim x$ since $UL(u_1)=u_2$ and $DL(u_1)=x$ are in $V(G)$. On the other hand, for the case of \ref{Fig_RelatedLocation_x-u1}(d) we get that $z=R(u_1)$. Then, the case is the same as the case of exchanging the positions of $u_2$ and $u_k$, and hence $z\thicksim x$. The other cases can be inspected by the symmetry and linearly convexity of graphs.

\begin{figure}[tp]
\begin{center}
\includegraphics[scale=1.0]{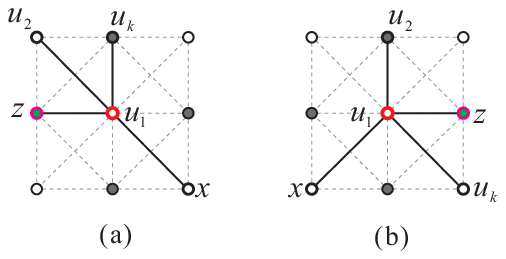}
\caption{The possible locations among $z, x, u_1, u_2, u_k$ for Claim 2 with $z \nsim x$.} \label{Fig_RelatedLocation_Claim2}
\end{center}
\end{figure}

We first consider the case of Fig. \ref{Fig_RelatedLocation_Claim2}(b). Since $DL(u_1)=x$ and $DR(u_1)=u_k$ are in $G$, we get $D(u_1)\in V(G)$ by Proposition \ref{ImmediateVertexInclude}. By symmetry, we can exchange $u_2$ and $u_k$. Replace $z=R(u_1)$ with $z=D(u_1)$. Then, $z \thicksim x$ and it can be proved by the similar arguments in proving Claim 1.

Now, we consider the case of Fig. \ref{Fig_RelatedLocation_Claim2}(a). We will construct a cycle $C'$ by the similar construction of Claim 1. Since $G$ is $2$-connected and linearly convex, at least one of $UR(u_1)$ and $DL(u_1)$ is in $G$. Suppose that $UR(u_1)$ is in $G$. By Proposition \ref{ImmediateVertexInclude}, $R(u_1)\in V(G)$. By symmetry, we can exchange the locations of $u_2$ and $u_k$. Let $z = R(u_1)$. Then, $z\thicksim x$, $z\thicksim u_2$, and it is the case of Fig. \ref{Fig_RelatedLocation_x-u1-z-Claim1}(a). Thus, an extended cycle $C'$ from $C$ and $x$ can be constructed.

In the following, assume that $UR(u_1)\not\in V(G)$. Then, $DL(u_1)\in V(G)$. Since $DL(u_1)\in V(G)$ and $DR(u_1)(=x)\in V(G)$, we get that $D(u_1)\in V(G)$ by Proposition \ref{ImmediateVertexInclude}. Let $y=D(u_1)$. Then, $y \thicksim z$, $y \nsim u_2$, and $y \nsim u_k$. By the case of Fig. \ref{Fig_RelatedLocation_x-u1-z-Claim1}(b) in Claim 1, $y\in C$. Let $y = u_t$. Then, $3\leqslant t\leqslant k-1$. By Proposition \ref{CycleVertexMerge_Obs}, $x \nsim u_{t-1}$ and $x \nsim u_{t+1}$. Then, $u_{t-1}, u_{t+1}\in \{L(y), DL(y), UL(y)(=z)\}$. Since $(u_1, u_2)\in C$, $z \thicksim u_1$, and $z \thicksim u_2$, we get that $z\in C$ by Proposition \ref{CycleVertexMerge_Obs}. Let $z = u_j$. Then, $3\leqslant j\leqslant k-1$. Suppose that $z=u_{t-1}$. Let $C_1 = u_1\rightarrow x\rightarrow u_t(=y)\rightarrow u_{t-1}(=z)\rightarrow u_{t-2}\rightarrow u_{t-3}\rightarrow \cdots \rightarrow u_2$ and let $P_1 = u_{t+1} \rightarrow u_{t+2}\rightarrow \cdots\rightarrow u_{k-1}\rightarrow u_k$. Then, $start(P_1)(=u_{t+1}) \thicksim u_t(=y)$ and $end(P_1)(=u_k) \thicksim u_{t-1}(=z)$. By Proposition \ref{CyclePathMerge_Obs}, $C_1$ and $P_1$ can be concatenated into a cycle $C'$ with $V(C') = V(C)\cup\{x\}$. On the other hand, suppose that $z=u_{t+1}$. Then, $C' = u_1 \rightarrow x \rightarrow u_t(=y) \rightarrow u_{t-1} \rightarrow u_{t-2} \rightarrow \cdots \rightarrow u_3 \rightarrow u_2 \rightarrow u_{t+1}(=z) \rightarrow u_{t+2} \rightarrow \cdots \rightarrow u_{k-1} \rightarrow u_k$ is a cycle covering $V(C)$ and $x$, and hence $C'$ is the extension of $C$. In the following, assume that $z \neq u_{t-1}$ and $z \neq u_{t+1}$. Then, $u_{t-1}, u_{t+1}\in \{L(y), DL(y)\}$ and, hence, $u_{t-1}\thicksim u_{t+1}$. Depending on whether $y$ appears before $z$ in $C$, there are two cases as follows.

\textit{Case} 1: $2 \leqslant j\leqslant t-2$. In this case, $y$ appears after $z$ in $C$. Suppose that $u_{j-1} \thicksim u_{j+1}$. We have that $u_{t-1} \thicksim u_{t+1}$, $z \thicksim u_2$, and $u_t(=y) \thicksim u_j(=z)$. Let $C' = u_1 \rightarrow x \rightarrow u_t(=y) \rightarrow u_j(=z) \rightarrow u_2 \rightarrow u_3 \rightarrow \cdots\rightarrow u_{j-1} \rightarrow u_{j+1} \rightarrow u_{j+2} \rightarrow \cdots \rightarrow u_{t-1} \rightarrow u_{t+1} \rightarrow u_{t+2} \rightarrow \cdots \rightarrow u_{k-1} \rightarrow u_k$. Then, $C'$ is a Hamiltonian cycle of $G[V(C)\cup\{x\}]$. In the following, suppose that $u_{j-1} \nsim u_{j+1}$. Then, $\{u_{j-1}, u_{j+1}\} = \{UL(z), DL(z)\}$ or $\{UL(z), D(z)\}$. We can see that $D(z)\in\{u_{t-1}, u_{t+1}\}$ and $1\geqslant |\{u_{j-1}, u_{j+1}\}\cap\{u_{t-1}, u_{t+1}\}| \geqslant 0$. Consider that $\{u_{j-1}, u_{j+1}\}\cap\{u_{t-1}, u_{t+1}\}=\emptyset$. Then, $\{u_{j-1}, u_{j+1}\} = \{UL(z), DL(z)\}$ and $\{u_{t-1}, u_{t+1}\} = \{L(y), DL(y)\}$. Let $C = u_1 \rightarrow u_2 \rightarrow P_1 \rightarrow z(=u_j) \rightarrow u_{j+1} \rightarrow P_2 \rightarrow y(=u_t) \rightarrow P_3 \rightarrow u_k$, where $end(P_1) = u_{j-1}$, $end(P_2)=u_{t-1}$, and $start(P_3)=u_{t+1}$. Depending on the related positions among $u_{j-1}, u_{j+1}, u_{t-1}, u_{t+1}$, there are four possible locations of vertices and subpaths in $C$, as shown in Fig. \ref{Fig_Claim2_Case1}(a)--(d). For instance, in Fig. \ref{Fig_Claim2_Case1}(a), $u_{j-1}=UL(z)$, $u_{j+1}=DL(z)$, $u_{t-1}=L(y)$, and $u_{t+1}=DL(y)$. For the cases of Fig. \ref{Fig_Claim2_Case1}(a)--(d), we can construct an extending cycle $C'$ from $C$ to cover $x$, as in Fig. \ref{Fig_Claim2_Case1}(e)--(h), respectively. For instance, consider the case of Fig. \ref{Fig_Claim2_Case1}(a), let $C' = u_1 \rightarrow x \rightarrow y \rightarrow P_3 \rightarrow u_k \rightarrow u_2 \rightarrow P_1 \rightarrow z \rightarrow u_{j+1} \rightarrow P_2$, as in Fig. \ref{Fig_Claim2_Case1}(e). Fig. \ref{Fig_Claim2_Case1}(f)--(h) depict the constructed cycles for Fig. \ref{Fig_Claim2_Case1}(b)--(d), respectively. On the other hand, consider that $\{u_{j-1}, u_{j+1}\}\cap\{u_{t-1}, u_{t+1}\}\neq\emptyset$. Then, $\{u_{j-1}, u_{j+1}\}\cap\{u_{t-1}, u_{t+1}\}=D(z)=L(y)$ and an extending cycle $C'$ can be constructed by the similar construction. %Note that edge $(DL(z), L(y))$ is in $C'$ if it is possible.

\begin{figure}[tp]
\begin{center}
\includegraphics[scale=1.0]{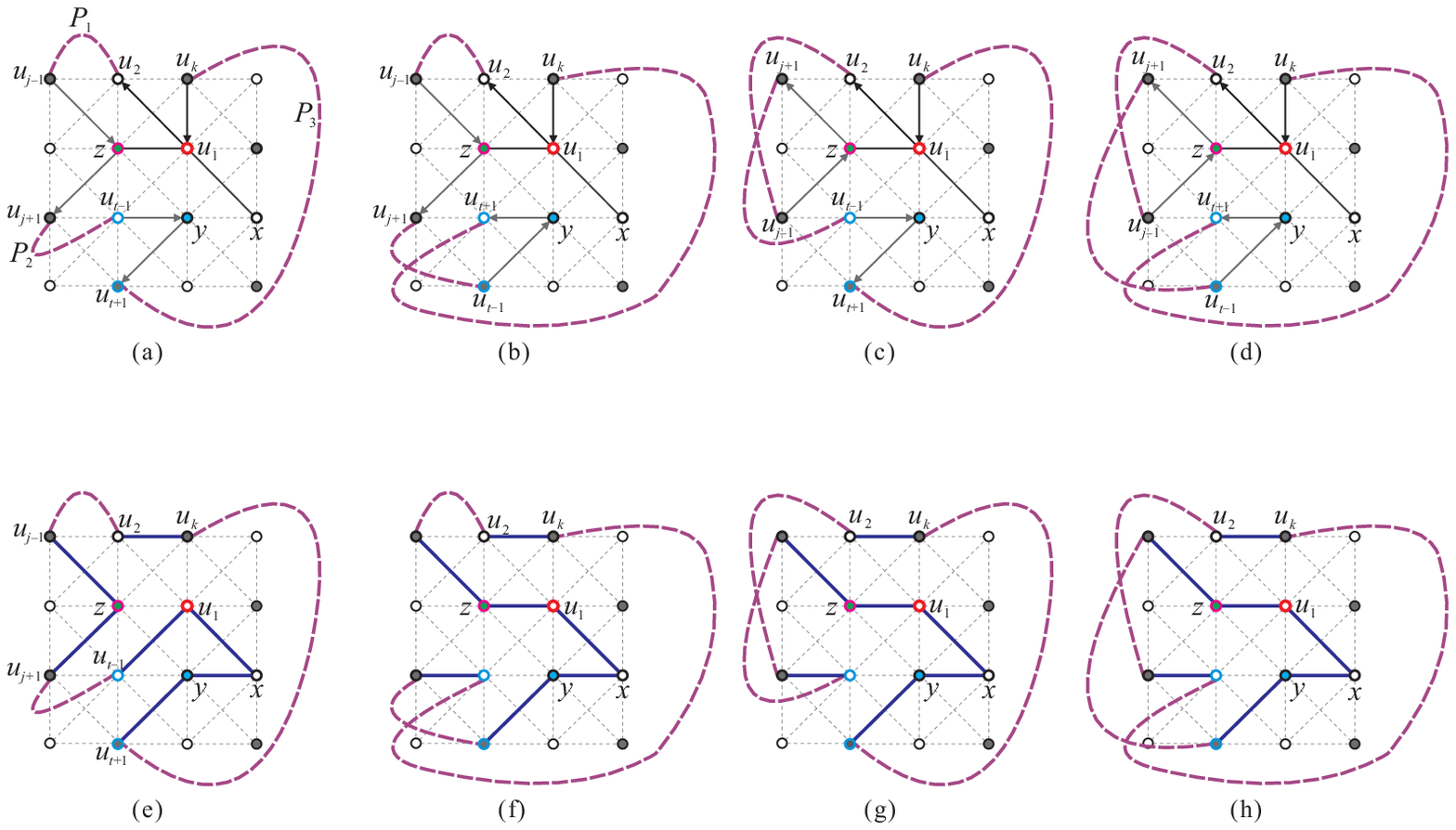}
\caption{(a)--(d) The possible cycle $C$ under that $z \nsim x$, $z\not\in \{u_{t-1}, u_{t+1}\}$, $2 \leqslant j\leqslant t-2$, $u_{j-1} \nsim u_{j+1}$, and $\{u_{j-1}, u_{j+1}\}\cap\{u_{t-1}, u_{t+1}\}=\emptyset$, where arrow solid lines indicate the edges in $C$ and dark dashed lines represent subpaths of $C$; (e)--(h) the constructed extending cycle $C'$ containing $V(C)$ and $x$ for (a)--(d), respectively.} \label{Fig_Claim2_Case1}
\end{center}
\end{figure}

\textit{Case} 2: $t+1 \leqslant j\leqslant k-1$. In this case, $y$ appears before $z$ in $C$. By similar arguments in proving Case 1, this case can be proved. Consider that $u_{j-1} \thicksim u_{j+1}$. Then, $C' = u_1 \rightarrow x \rightarrow u_t(=y) \rightarrow u_j(=z) \rightarrow u_2 \rightarrow u_3 \rightarrow \cdots\rightarrow u_{t-1} \rightarrow u_{t+1} \rightarrow u_{t+2} \rightarrow \cdots \rightarrow u_{j-1} \rightarrow u_{j+1} \rightarrow u_{j+2} \rightarrow \cdots \rightarrow u_{k-1} \rightarrow u_k$ is a cycle extended from $C$ to cover $x$. Suppose that $u_{j-1} \nsim u_{j+1}$ below. By the same arguments in proving Case 1, we first consider that $\{u_{j-1}, u_{j+1}\}\cap\{u_{t-1}, u_{t+1}\}=\emptyset$. The possible cycles $C$ and $x$ are shown in Fig. \ref{Fig_Claim2_Case2}(a)--(d). Fig. \ref{Fig_Claim2_Case2}(e)--(h) depict the constructed cycles for the cases of Fig. \ref{Fig_Claim2_Case2}(a)--(d), respectively. On the other hand, consider that $\{u_{j-1}, u_{j+1}\}\cap\{u_{t-1}, u_{t+1}\}\neq\emptyset$. Then,  $\{u_{j-1}, u_{j+1}\}\cap\{u_{t-1}, u_{t+1}\}=D(z)=L(y)$. In our construction of cycles in Fig. \ref{Fig_Claim2_Case2}(e)--(g), they contain edge $(DL(z), L(y))$. Thus, we can view this edge as one vertex $D(z)$ and construct the extending cycle $C'$ for the case of $\{u_{j-1}, u_{j+1}\}\cap\{u_{t-1}, u_{t+1}\}\neq\emptyset$. By similar technique, it can be applied to the construction of $C'$ for Fig. \ref{Fig_Claim2_Case2}(h) under that $\{u_{j-1}, u_{j+1}\}\cap\{u_{t-1}, u_{t+1}\}=\{D(z)=L(y)\}$.

\begin{figure}[tp]
\begin{center}
\includegraphics[scale=1.0]{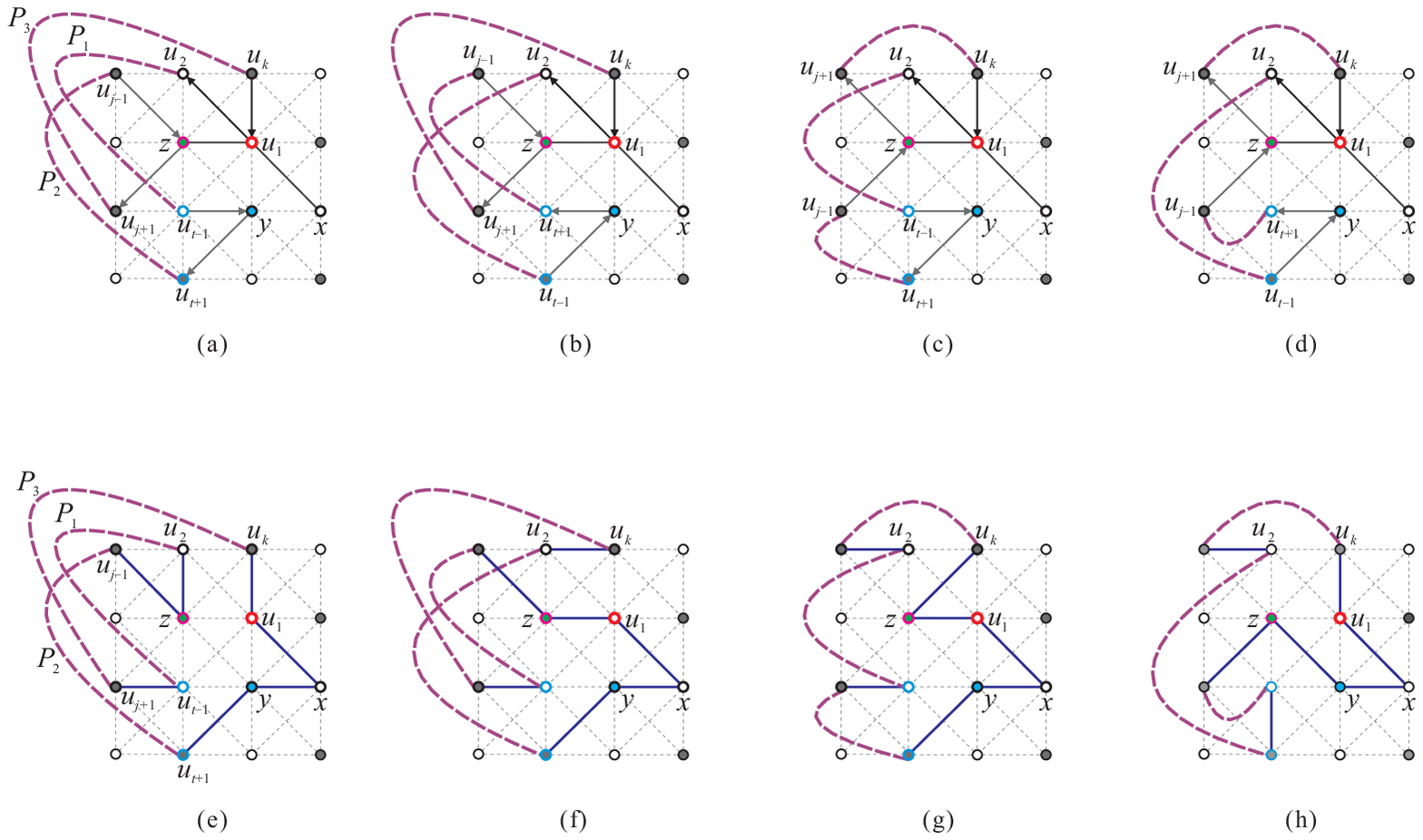}
\caption{(a)--(d) The possible cycle $C$ under that $z \nsim x$, $z\not\in \{u_{t-1}, u_{t+1}\}$, $t+1 \leqslant j\leqslant k-1$, $u_{j-1} \nsim u_{j+1}$, and $\{u_{j-1}, u_{j+1}\}\cap\{u_{t-1}, u_{t+1}\}=\emptyset$, where arrow solid lines indicate the edges in $C$ and dark dashed lines represent subpaths of $C$; (e)--(h) the constructed extending cycle $C'$ containing $V(C)$ and $x$ for (a)--(d), respectively.} \label{Fig_Claim2_Case2}
\end{center}
\end{figure}

We have considered all cases under $z \nsim x$ to show that $C$ is cycle extendable. This completes the proof of Claim 2.\hfill$\square$\\

It immediately follows from Theorem \ref{CycleExtendable} that the following theorem holds true.

\begin{thm}\label{Hamiltonian}
Let $G$ be a $2$-connected, linear-convex supergrid graph. Then, $G$ contains a Hamiltonian cycle.
\end{thm}

%%% ----------------------------------------------------------------------
\section{Concluding remarks}\label{Conclusion}
%%% ----------------------------------------------------------------------

The Hamiltonian cycle and Hamiltonian path problems for supergrid graphs were known to be NP-complete. In this paper, we first show that $2$-connected, linear-convex supergrid graphs are locally connected. We then prove that $2$-connected, linear-convex supergrid graphs are cycle extendable and hence are Hamiltonian. It is interesting to see whether the Hamiltonian problems for the other subclasses of supergrid graphs, including solid and locally connected, are polynomial solvable. We would like to post it as an open problem to interested readers.

%\section*{Acknowledgments}
%The author gratefully acknowledges the helpful comments and
%suggestions of the reviewers, which have improved the presentation
%and have strengthened the contribution. This work is partly
%supported by the National Science Council of Republic of China
%under grant no. NSC 102-2221-E-324-xxx-MY2.

%%% ----------------------------------------------------------------------
%%% Reference
%%% ----------------------------------------------------------------------

\end{document}